\documentclass[journal]{IEEEtran}
\ifCLASSINFOpdf
\else
\fi
\usepackage{cite}
\usepackage{graphicx}
\usepackage{caption}
\usepackage{subcaption}
\usepackage{float}
\usepackage{csquotes}
\usepackage{array}
\newcolumntype{P}[1]{>{\centering\arraybackslash}p{#1}}
\newcolumntype{M}[1]{>{\centering\arraybackslash}m{#1}}
\usepackage{multicol}
\usepackage{amsmath}
\usepackage{bm}
\usepackage{amssymb}
\usepackage{enumitem}
\usepackage{color}
\usepackage{algorithmic}
\usepackage{amsthm}

\newtheorem{theorem}{Theorem}
\newtheorem{assumption}{Assumption}
\newtheorem{lemma}{Lemma}
\newtheorem{corollary}{Corollary}
\newtheorem{remark}{Remark}
\newtheorem{example}{Example}

\begin{document}
\newcommand{\ggll}{\mathrel{\substack{>\\[-.05em]<}}}
\title{Low-complexity Distributed Detection with One-bit Memory Under Neyman-Pearson Criterion}

\author{Guangyang Zeng,~\IEEEmembership{Student Member,~IEEE,}
       Xiaoqiang Ren,~\IEEEmembership{Member,~IEEE,}
        and~Junfeng~Wu,~\IEEEmembership{Senior Member,~IEEE}% <-this % stops a space
\thanks{Some preliminary results of the current manuscript were briefly presented at the 56th Annual Allerton Conference on Communication, Control, and Computing in Urbana-Champaign, USA, 2018~\cite{zeng2018novel}.}% <-this % stops a space
\thanks{G. Zeng and J. Wu are with the College of Control Science and Engineering, State Key Laboratory of industrial control and technology, Zhejiang University, P.R.China. Email: \{gyzeng,jfwu\}@zju.edu.cn.}
\thanks{X. Ren is with the School of Mechatronic Engineering and Automation, Shanghai
University, Shanghai, P.R. China, Email:xqren@shu.edu.cn.}}
\maketitle

\begin{abstract}
We consider a multi-stage distributed detection scenario, where $n$ sensors and a fusion center (FC) are deployed to accomplish a binary hypothesis test. At each time stage, local sensors generate binary messages, assumed to be spatially and temporally independent given the hypothesis, and then upload them to the FC for global detection decision making. We suppose a one-bit memory is available at the FC to store its decision history and focus on developing iterative fusion schemes. We first visit the detection problem of performing the Neyman-Pearson (N-P) test at each stage and give an optimal algorithm, called the oracle algorithm, to solve it.
Structural properties and limitation of the fusion performance in the asymptotic regime are explored for the oracle algorithm. We notice the computational inefficiency of the oracle fusion and propose a low-complexity alternative, for which the likelihood ratio (LR) test threshold is tuned in connection to the fusion decision history compressed in
the one-bit memory. The low-complexity algorithm greatly brings down the computational complexity at each stage from $O(4^n)$ to $O(n)$. We show that the proposed algorithm is capable of converging exponentially to the same detection probability as that of the oracle one.
Moreover, the rate of convergence is shown to be asymptotically identical to that of the oracle algorithm. Finally, numerical simulations and real-world experiments demonstrate the effectiveness and efficiency of our distributed algorithm.

\end{abstract}

\begin{IEEEkeywords}
Distributed detection; Binary detection; Neyman-Pearson; Computational complexity
\end{IEEEkeywords}

% \fi
\IEEEpeerreviewmaketitle

\section{Introduction} \label{Introduction}
\IEEEPARstart{W}{ith} the booming development of wireless sensor networks, distributed detection has gained its flourishing as it shows great advantages on expanding space coverage, increasing reliability, enhancing system viability, etc.~\cite{tenney1981detection,chair1986optimal,tsitsiklis1989decentralized,viswanathan1997distributed,blum1997distributed,chamberland2007wireless,varshney2012distributed,drakopoulos1991optimum,kam1992optimal}. For a literature review, one may refer to~\cite{viswanathan1997distributed,blum1997distributed,chamberland2007wireless}. We consider the parallel distributed detection architecture, in which there are a number of local sensors and a fusion center (FC). The sensors are spatially distributed and supposed to have computation capacity by which local decisions are made according to their observations. These local decisions are then transmitted to the FC, based on which the global decision is made following some fusion rules.

Distributed detection problems are generally to design rules, including the local decision rules and the fusion rule, through which the system performance, with respect to some objectives, is maximized. Most literature has assumed the observations among sensors are conditionally independent given each hypothesis for tractability~\cite{drakopoulos1991optimum,kam1992optimal,yan2001distributed,chamberland2006dense,khalid2011cooperative,veeravalli2012distributed}. Otherwise the problem is generally intractable, e.g., it is NP-hard under the Neyman-Pearson (N-P) formulation~\cite{tsitsiklis1985complexity}. Under the conditional independence assumption, the optimal decision rules at the sensors as well as at the FC are threshold rules based on likelihood ratios (LR) under both Bayesian and N-P formulations~\cite{viswanathan1997distributed}. We note that although the optimality of LR tests has been established, calculating their optimal thresholds is usually computationally complicated. Even if the decision rules at sensors have been fixed, the computational cost to find an optimal fusion rule is exponentially increasing with respect to sensor number~\cite{tsitsiklis1989decentralized}. 
Hence, reducing computational complexity is a widely studied topic in distributed detection. To this aim, the asymptotic regime for which a large number of sensors are involved is a good choice~\cite{tsitsiklis1988decentralized,chamberland2004asymptotic,kreidl2010decentralized,dong2017detection,tay2012value}, as in this case, if the sensors are i.i.d., it is asymptotically optimal to adopt the same decision rule at the local sensors, based on which the optimal fusion strategy yields the voting rule. In these works, the optimal error exponent is derived as Chernoff information in the Bayesian formulation or KL divergence in the N-P formulation. There are also some suboptimal methods that pursue suboptimal solutions with relatively low complexity~\cite{hoballah1989distributed,zhang2002optimal,quan2009optimal,rahaman2017low}. Although they may be optimal in some artificially specified conditions, e.g., Quan \emph{et al.}~\cite{quan2009optimal} gave the optimal \emph{linear} fusion strategy, they are not optimal in the global sense. 

In this paper, we are also devoted to devising an efficient algorithm for the scenario of interest. A multi-stage distributed detection problem, which consists of multiple time stages, is considered. The system diagram is shown in Fig.~\ref{System_architecture}. At the $k$-th stage, each sensor will send a quantized message $u_i^k$ (binary-valued in this paper) to the FC at which the binary detection result $u_0^k$ is produced to decide between two hypotheses ($H_0$ and $H_1$)\footnote{We require that the FC gives its decision at every stage to support the potential subsequent course of action or higher-level fusion procedure.}. Using the N-P criterion, we constrain that the false alarm probability of $u_0^k$ cannot exceed a prescribed value. To alleviate system complexity, we will make two efforts. On the one hand, we suppose that a one-bit memory is available at the FC to compress historical sensor messages. In this manner, both computational and memory complexity will not increase with the time stage going large. Nevertheless, we will show (in Section~\ref{oracle_optimal_algorithm}) that the calculation of the optimal thresholds is computationally inefficient. Hence, on the other hand, we are also devoted to seeking a low-complexity LR threshold generation method, based on which the detection performance of the modified algorithm does not degrade much. 
\begin{figure}[htbp]
	\centering
	\includegraphics[width=0.48\textwidth]{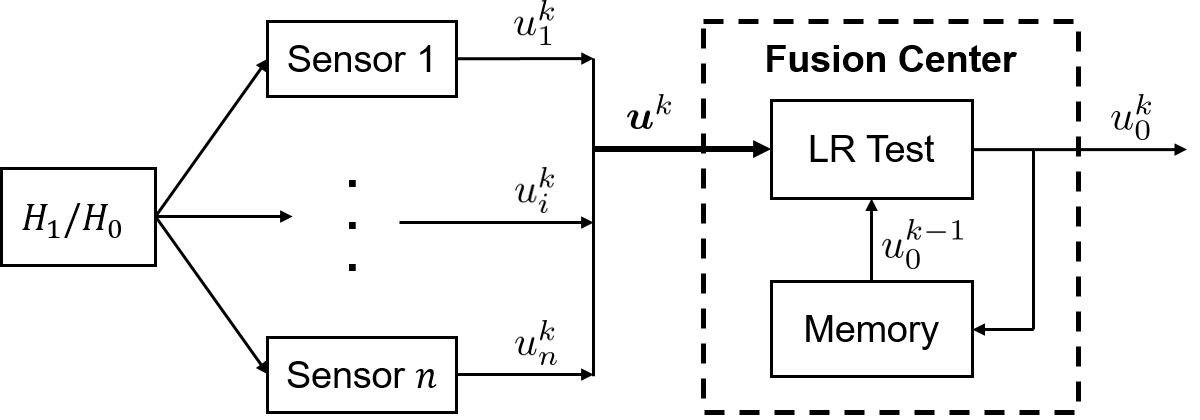}
	\caption{System diagram of distributed  hypothesis testing in a sensor network. The sensors transmit quantized messages to the FC at each time. The FC makes a global decision in an online manner, fusing its previous decision and the sensors' local messages.}
	\label{System_architecture}
\end{figure} 

\subsection{Related literature}
A series of works by Cover \emph{et al.}~\cite{cover1969hypothesis,hellman1970learning,cover1976optimal} investigated hypothesis testing with finite memory where the $m$-valued ($m \geq 2$) statistic at present is recursively updated based on the statistic of the last stage and the present observation. The properties of the asymptotic\footnote{In the subsequence of this paper, without specific illustration, ``asymptotic'' means that the time stage goes to infinity.} probability of error under different cases (unbounded and bounded LR) are characterized. It is noteworthy that these works focus on the asymptotic case. Therefore the transient constraint on the false alarm probabilities cannot be guaranteed. Moreover, due to the Bayesian formulation, the asymptotic false alarm probability may also fail to satisfy the constraint in the case of bounded LR. 
Sequential detection is a widely studied problem, which consists of multiple stages. In most of the literature, e.g.,~\cite{veeravalli1993decentralized,krishnamurthy2011bayesian,leonard2018robust}, the detection result is made only when some criterion is satisfied (which is called the stopping time). The general result is a random wait for the final decision about the true phenomenon.
We note that the system diagram in Fig.~\ref{System_architecture} can be converted into a serial (or tandem) detection structure. In serial topology~\cite{viswanathan1988optimal,tang1991optimization,tarighati2016decentralized}, the $k$-th sensor makes a decision based on its own observation and the decision of the $(k-1)$-th sensor. In these problems, the N-P criterion is applied to the last sensor which gives the final decision. In order to calculate the optimal threshold at each sensor, it is assumed that the observations are continuously distributed~\cite{tang1991optimization}. Thus, applying the N-P criterion at each stage where the observation of the FC $\bm u^k=[u_1^k,\ldots,u_n^k]$ is discrete (in this paper) can be viewed as a supplement to the canonical serial distributed detection. 
Using the previous decision to improve the current one resembles the ``unlucky broker problem''~\cite{marano2010refining,marano2010bayesian}. However, in that problem, no new observation is available, and the decision maker needs to refine the test by exploiting the previous decision and part of the historical observation. 
Kam \emph{et al.}~\cite{kam1999performance} investigated the same system diagram as ours where a one-bit memory is supposed at the FC to compress historical sensor messages. At each stage, the Bayesian criterion is applied to minimize the probability of error. The gap between the Bayesian setting and the N-P one is nontrivial in terms of both problem formulation and analysis method. When the apriori probability of the underlying hypothesis and the cost of each course of action are unknown, or there is a constraint on the false alarm probability, the N-P formulation is required. We note that the performance analysis is conducted based on the geometric interpretation of the iteration functions of the detection and false alarm probabilities in~\cite{kam1999performance}, while we utilize the property of the receiver's operating characteristic (ROC) curve of the canonical memoryless system. Moreover, the asymptotic properties between the two problems are also different (see Remark~\ref{asymptotic_comparison}). 

\subsection{Main results} 
Considering the system diagram in Fig.~\ref{System_architecture}, we require that the FC gives its binary decision at every stage to support the potential subsequent course of action or higher-level fusion procedure. The N-P criterion is applied to set a prescribed bound for the transient false alarm probability at each stage. We assume that only a one-bit memory is available in the FC to compress local messages for the following considerations. From the perspective of practical applications, in the scenario where the FC serves as a relay in a multi-hop network or a hierarchical network, it is expected that the resources, including memory and computation capacity, are constrained in the FC. If the FC is battery-powered, the one-bit memory setting with relatively low computational complexity can increase the battery life. In addition, theoretically speaking, to the best of our knowledge, the one-bit memory case has not been studied yet in our problem setup. We believe that the full characterization of the one-bit problem, which is the extreme case of finite memory setting, will promote the investigation of the multi-bit memory situation. 
Based on the one-bit memory architecture, in Section~\ref{oracle_optimal_algorithm}, we give the optimal algorithm through which the detection probability is sequentially maximized at each stage. We note that the transient detection probability is not analytically available, then based on some results in~\cite{xiang2001performance},~\cite{zhao2007performance}, we derive the asymptotic detection probability. Nevertheless, the oracle optimal algorithm is computationally inefficient, which has an exponential (in terms of the number of sensors) computational complexity at each stage. 
Therefore, we further turn to optimize the asymptotic detection probability instead of the transient one and propose an asymptotically optimal algorithm with low complexity. 
The low-complexity algorithm was first proposed in our previous work~\cite{zeng2018novel}, where the one-bit memory is used to adaptively select the thresholds of the LR tests. We have shown that the detection and false alarm probabilities of the FC converge. In addition, the optimal alternative thresholds have been derived for the case of two homogeneous sensors\footnote{For homogeneous sensors, their outputs have identical distribution given each hypothesis, which does not hold for heterogeneous sensors.}. However, the problem of performing an N-P test at every stage to optimize transient detection metrics has not been considered, and the development of theoretical counterparts for heterogeneous sensor networks of any size is absent. In this paper, we are going to tackle these problems. In summary, the main contributions are as follows:
\begin{enumerate}
\item [$(i).$] We visit the problem of performing an N-P test at each stage with a one-bit memory and devise an optimal algorithm, called the oracle algorithm, to solve it. Then we focus on analyzing the value of the additional one-bit information. Since the fusion threshold and the detection probability at each stage are not available analytically, we turn to study the asymptotic detection performance, proving the convergence of the detection probability and deriving its asymptotic value (Theorem~\ref{oracle_performance}).

\item [$(ii).$] We notice the computational inefficiency of the oracle fusion and propose a low-complexity fusion policy, greatly bringing down the computational complexity at each stage from $O(4^n)$ to $O(n)$. In this algorithm, the threshold of the LR test is selected from two pre-calculated values based on the fusion decision history compressed in the one-bit memory. We show that by using the proposed algorithm, the detection probability converges at an exponential rate. Moreover, it achieves exactly the same asymptotic detection probability as that of the oracle fusion by proper parameter selection (Theorem~\ref{possible_t1}).

\item [$(iii).$] We discuss the relations between the proposed and the oracle optimal algorithms. The devising of the low-complexity detection rule is enlightened from the converging of the oracle test rule. The oracle algorithm converges to a stationary rule that is identical in form to the proposed low-complexity one. From this perspective, the proposed policy can be treated as an approximation of the oracle one (see Remark~\ref{relationship_of_algorithms}). In addition, we prove that both algorithms converge exponentially at an asymptotically identical rate (Theorem~\ref{convergence_sped}).
\end{enumerate}

The rest of the paper is organized as follows. In Section~\ref{Architecture_problem_formulation}, we introduce the system architecture and formulate our problem. In Section~\ref{preliminaries}, we give some preliminaries on the classical one-stage distributed detection, which facilitates the subsequent multi-stage analysis. In Section~\ref{oracle_optimal_algorithm}, we give the oracle optimal algorithm to sequentially maximize the transient detection probability at each stage and derive its asymptotic value. In Section~\ref{the_proposed_algorithm}, we propose a low-complexity algorithm to maximize the asymptotic detection probability instead of the transient one and make a thorough comparison with the oracle optimal algorithm. Simulation and real experiment results are presented in Section~\ref{simulation_experiment_results}, followed by remarks in Section~\ref{conclusion}.

\section{Problem formulation} \label{Architecture_problem_formulation}
\subsection{System Architecture}
The system diagram considered in this paper is shown in Fig.~\ref{System_architecture}. We regard an event detection problem as a binary hypothesis test problem, where $H_0$ and $H_1$ denote the null hypothesis (i.e., the event is absent) and the alternative hypothesis (i.e., the event is present), respectively. There are 
 $n$ sensors making observations of a common real event $H$ over time. Throughout this paper, we make the following assumptions:
\begin{assumption} \label{unchanged_true_hypothesis}
	The true common event $H$ remains unchanged over time. 
\end{assumption} 
The sensors are all binary sensors. At the $k$-th time stage, the $i$-th sensor produces a binary message $U_i^k$ with\footnote{We use the superscript $k$ to denote the $k$-th stage. For the $k$-th power function, we will add the brackets, i.e., $()^k$.}
\begin{equation} \label{local_performance}
\begin{split} 
& P(U_i^k=1 \mid H_1)=p_i,~P(U_i^k=0 \mid H_1)=1-p_i, \\
& P(U_i^k=1 \mid H_0)=q_i,~P(U_i^k=0 \mid H_0)=1-q_i,
\end{split} 
\end{equation}
where $p_i,q_i \in (0,1)$ are known \emph{a priori}.
\begin{assumption}  \label{conditional independence}
Under either hypothesis, the random variables $U_i^k$ are spatially and temporally independent. 
\end{assumption}
\begin{remark} \label{binary_sensor}
	We refer to any sensor whose output is binary as a binary sensor. It can be the case that the observation of the sensor itself is binary. Or the observation may be arbitrary, and a local decision rule is applied to produce a binary result. 
\end{remark}
At each time stage $k$, sensor $i$ obtains a realization $u_i^k$ of the random variable $U_i^k$, and sends the message to a FC through an error-free channel. The FC receives local messages and fuses them into a global decision in an online manner, regarding the presence or absence of an event in the area. In addition, the FC has a one-bit memory to store the previous global decision $u_0^{k-1}$, and then at time $k$ it makes a new online decision $u_0^k=\Gamma^k(u_1^k,\ldots,u_n^k,u_0^{k-1}) \in \{0,1\}$, where $\Gamma^k:\{0,1\}^{n+1} \rightarrow \{0,1\}$ is referred to as the fusion rule of the FC. 
%The objective is to maximize or minimize some functions, e.g., the N-P criterion, by optimizing the decision rules of local sensors and the FC. 

\subsection{Problem of Interest}

In this paper, the FC is supposed to give a binary decision at every stage for some potential subsequent tasks. We consider the N-P criterion, with the purpose of maximizing the detection probability at each stage under the constraint that the false alarm probability does not exceed a prescribed bound. Formally, the problem is formulated as follows: find a sequence of fusion rules $\bm{\Gamma}=\{\Gamma^1,\Gamma^2,\ldots\}$ such that the following problem is solved for each time $k$:
\begin{equation} \label{N-P_test_original}
\begin{aligned} 
\mathop{\rm{maximize}}_{\bm \Gamma} \quad & p_0^k(\bm \Gamma) \\
{\rm subject~to} \quad & q_0^k(\bm \Gamma) \le \alpha,
\end{aligned}
\end{equation}
where $p_0^k$ and $q_0^k$ are the detection probability and false alarm probability of $u_0^k$, respectively, and $\alpha \in [0,1]$ is the prescribed bound on the false alarm probability. {\it The detection probability is the probability that $H_1$ is declared when the true hypothesis is $H_1$, and the false alarm probability is the probability that $H_1$ is declared when the true hypothesis is $H_0$.} For problem~\eqref{N-P_test_original}, we focus on the following two issues:
\begin{enumerate}
	\item [$(i).$] We are interested in the optimal fusion rules $\bm \Gamma$ for problem~\eqref{N-P_test_original} under the system diagram with a one-bit memory in Fig.~\ref{System_architecture}. In particular, we wonder how much detection performance improvement (compared with memoryless system structure) can be made by the one-bit memory. 
	
	\item [$(ii).$] We are interested in seeking a low-complexity fusion algorithm whose detection probability does not degrade much compared to the optimal algorithm (e.g., it owns the asymptotically optimal property) while satisfying the constraint on the false alarm probability.
\end{enumerate}

Before proceeding, we will first introduce in Section~\ref{preliminaries} some preliminaries on the canonical one-stage distributed detection under the N-P formulation, based on which we will further investigate the multi-stage distributed detection in Sections~\ref{oracle_optimal_algorithm} and~\ref{the_proposed_algorithm}.

\section{Preliminaries: Canonical N-P Distributed Detection} \label{preliminaries}
We introduce basic results in the canonical N-P distributed detection. We consider a sensor network consisting of $n$ sensors and a FC. Each sensor $i$ makes a binary decision $u_i$ and sends it to the FC. The fusion rule $\Gamma$ at the FC is a mapping from $[u_1,\ldots,u_n]$ to $0$ or $1$. Under the N-P criterion, the global detection probability $p_0 $ is to be maximized with respect to the fusion rule, subject to a prescribed global false alarm probability $q_0$: 
\begin{equation} \label{N-P_test_one_stage}
\begin{aligned} 
\mathop{\rm maximize}_{\Gamma} \quad & p_0(\Gamma) \\
{\rm subject~to} \quad & q_0(\Gamma) \le \alpha.
\end{aligned} 
\end{equation}
\begin{remark} \label{one_multiple_stage}
In multi-stage  distributed detection, at each time problem~\eqref{N-P_test_original} solves a problem in the form of~\eqref{N-P_test_one_stage}. In other words, a multi-stage distributed detection problem~\eqref{N-P_test_original} comprises a series of one-stage problems.
\end{remark}
The N-P lemma~\cite{van2004detection} points out that the optimal fusion rule of~\eqref{N-P_test_one_stage} is in terms of the likelihood ratio (LR) test:
\begin{equation} \label{LR_test_0_one_stage}
\Lambda:=  \frac{P(\bm {u}|{H_1})}{P(\bm {u}|{H_0})}=
\frac{\prod_{i=1}^{n} P(u_i|H_1)}{\prod_{i=1}^{n} P(u_i|H_0)}
\mathop {\ggll}\limits_{H_0}^{{H_1}}  (t,\lambda), 
\end{equation}
where $\Lambda$ is called the likelihood ratio, $\bm {u}=[u_1,\ldots,u_n]$, and $\Lambda \mathop {\ggll}\limits_{H_0}^{{H_1}} (t,\lambda)$ means that:
\begin{equation}\label{greater_less}
u_0=
\begin{cases}
1, &\text{if } \Lambda>t, \\
0, &\text{if } \Lambda<t,\\
1 ~\text{with probability} ~\lambda, &\text{if } \Lambda=t, \\
\end{cases}
\end{equation}
with $u_0$ being the global decision. In~\eqref{LR_test_0_one_stage} the equality holds due to Assumption~\ref{conditional independence}. Consequently, problem~\eqref{N-P_test_one_stage} can be equally recast as the following optimization problem: 
\begin{equation} \label{N-P_test_onestage}
\begin{aligned} 
\mathop{\rm maximize}_{t\in \mathbb R_+,\lambda\in [0,1]} \quad & p_0(t, \lambda) \\
{\rm subject~to} \quad & q_0(t, \lambda) \le \alpha,
\end{aligned}
\end{equation}
where $\mathbb R_+$ stands for the set of positive reals. Since $u_i$'s are all binary, the likelihood ratio $\Lambda$ can have at most $2^n$ possible values. The distributions of $\Lambda$ conditioned on $H_0$ and $H_1$ are illustrated in Fig.~\ref{Probability_LR}. Let the finite set $\bf \Lambda$ denote the assemble of the possible values of $\Lambda$. Then based on the LR test~\eqref{LR_test_0_one_stage}, the detection and false alarm probabilities can be calculated via the relation:
\begin{equation}\label{detection_probability_calculation}
p_0(t, \lambda)=\sum_{\Lambda \in {\bf \Lambda}: \Lambda>t} P(\Lambda|H_1)+\lambda \sum_{\Lambda \in {\bf \Lambda}: \Lambda=t} P(\Lambda|H_1), 
\end{equation}
and 
\begin{equation}\label{falsealarm_probability_calculation}
q_0(t, \lambda)=\sum_{\Lambda \in {\bf \Lambda}: \Lambda>t} P(\Lambda|H_0)+\lambda \sum_{\Lambda \in {\bf \Lambda}: \Lambda=t} P(\Lambda|H_0). 
\end{equation}

It is well-known that $p_0$ reaches its maximum when $q_0 = \alpha$~\cite{van2004detection}, i.e., to optimally solve the problem~\eqref{N-P_test_onestage}, we need to obtain a pair of $(t,\lambda)$ satisfying $q_0(t, \lambda)=\alpha$. However, this procedure is quite tedious, especially when the number of sensors is large. 
What we  do first is to sort the likelihood ratio $\Lambda$ in non-decreasing order, denoted as $\Lambda_1 \le\Lambda_2 \le \dots \le \Lambda_{2^{n}}$ without loss of generality. Then we find the smallest index $j $ such that
\begin{equation} \label{find_optimal_threshold}
\sum_{\Lambda \in {\bf \Lambda}: \Lambda>\Lambda_j} P(\Lambda|H_0) \le \alpha
\end{equation}and denote the corresponding $\Lambda_j$ as $\Lambda_*$. Finally, the optimal threshold $t$ and the random factor $\lambda$ are computed as follows:
\begin{equation} \label{oracle_threshold}
\begin{cases}
t=\Lambda_*, \\
\lambda=\frac{\alpha-\sum_{\Lambda \in {\bf \Lambda}: \Lambda>\Lambda_*} P(\Lambda|H_0)}{P(\Lambda_*|H_0)}.
\end{cases}       
\end{equation}In general, we do not have an explicit expression for the optimal $p_0$ in terms of $\alpha$, while alliteratively we can have an implicit form via~\eqref{oracle_threshold} and~\eqref{detection_probability_calculation}.

Before proceeding, we introduce a metric for each sensor~\cite{xiang2001performance}:
\begin{equation}\label{new_performance_metric}
R_i=\frac{p_i}{1-p_i}\frac{1-q_i}{q_i}.
\end{equation}
\begin{remark}
We notice that the $R_i$ in~\eqref{new_performance_metric} can be read as an odds ratio~\cite{szumilas2010explaining}. It reflects how sensitive a detection decision maker in question is to the two different hypothesises $H_0$ and $H_1$ in environment. When $R_i<1$, i.e., $p_i<q_i$, it means that the local detector at sensor $i$ is counterproductive, and vice versa.  While $R_i=1$ implies that sensor $i$ is blind to the binary hypothesis.
\end{remark}
In this paper, we consider all the local decision makers as productive by making Assumption~\ref{productive_sensors}. Actually, if some of the sensors are counterproductive, we can make them be productive by flipping their outputs of $0$ and $1$.

\begin{assumption} \label{productive_sensors}
For each sensor $i$, we assume that $R_i>1$, that is, $p_i>q_i$.
\end{assumption}
In the binary detection problem, the ROC curve is a graphical representation that plots the detection probability versus the false alarm probability at different discrimination settings ($t$ and $\lambda$ in our problem). The following lemma is on properties of the ROC curve of the fusion result:
\begin{lemma}\label{ROC_curve_one_stage}
Consider the FC implementing~\eqref{LR_test_0_one_stage} with the optimal parameters~\eqref{oracle_threshold}. For a fusion of $n$ sensors, we have
	\begin{enumerate}
	\item [$(1)$.] The ROC curve of the FC consists of $2^n$ segments.
	\item [$(2)$.] Denote the slopes of $2^n$ segments in ROC from left to right as $l_1,\ldots,l_{2^n}$, then  $l_1>l_2 \geq \dots \geq l_{2^n-1}>l_{2^n}$.
\end{enumerate}
\end{lemma}
The proof of Lemma~\ref{ROC_curve_one_stage} is presented in Appendix~\ref{appendix_ROC_curve_one_stage}.

\begin{figure*}[htbp]
	\centering
	\begin{subfigure}{0.8\columnwidth}
		%		\centering
		\includegraphics[width=\columnwidth]{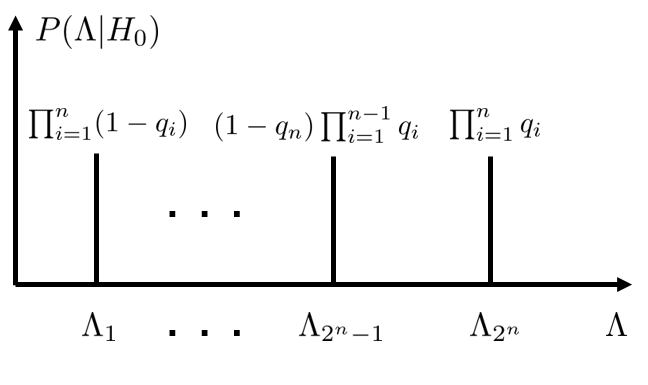}%
		\caption{Distribution of $\Lambda$ conditioned on $H_0$}
		\label{Probability_LR_H0}
	\end{subfigure}
	\hspace{0.1\textwidth}
	\begin{subfigure}{0.8\columnwidth}
		\includegraphics[width=\columnwidth]{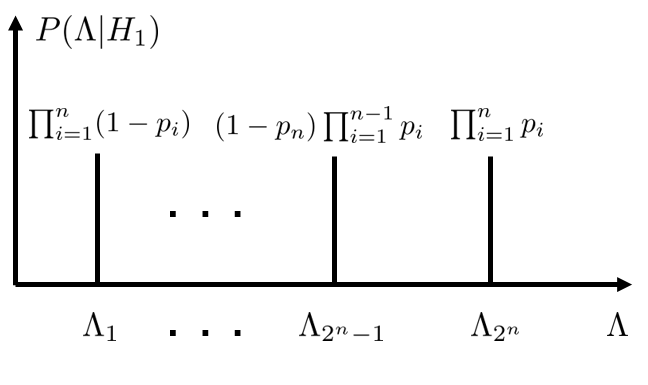}%
		\caption{Distribution of $\Lambda$ conditioned on $H_1$}
		\label{Probability_LR_H1}
	\end{subfigure}\hfill%
	\caption{Distribution of $\Lambda$ in single-stage decision fusion: For $n$ sensors, there are $2^n$ possible likelihood ratio values, denoted as $\Lambda_1 \leq \Lambda_2 \leq \dots \leq \Lambda_{2^n}$. There may be some likelihood ratios having the same value. In this sketch figure, they are separated and assigned probabilities respectively.}
	\label{Probability_LR}
\end{figure*}

Next, we will introduce a lemma on comparison of the detection performances before and after fusion:
\begin{lemma}[Theorem~2~\cite{zhao2007performance}] \label{better_performance_2}
	Consider problem~\eqref{N-P_test_onestage} in a configuration of $n \geq 3$ sensors with $R_1\geq R_2\geq \dots \geq R_n$. By using~\eqref{oracle_threshold}, the following statements are true:
	\begin{enumerate}
		\item [$(1)$.] Let $\alpha=q_1$, we have $p_0>p_1$ if $R_1<\prod_{i=2}^{n} R_i$ and $p_0=p_1$ if $R_1 \geq \prod_{i=2}^{n} R_i$.
		\item [$(2)$.] Let $\alpha=q_i, i \in \{2,\ldots,n\}$, we have $p_0>p_i$.
	\end{enumerate}
\end{lemma}

For a better understanding of the above two lemmas, we give an example of a sensor network consisting of three sensors.
\begin{example}\label{example:ROC_four_detectors}
Consider the case of three sensors with the following parameters:
	\begin{align*}
		& p_1=0.74,~p_2=0.66,~p_3=0.61, \\
		& q_1=0.16,~q_2=0.32,~q_3=0.39.
	\end{align*}These parameters reflect the detection performance of the three sensors respectively. A FC fuses local decisions from the sensors under a N-P criterion, and its detection and false alarm probabilities can be computed by~\eqref{detection_probability_calculation} and~\eqref{falsealarm_probability_calculation}.
The ROCs of the the four detectors are plotted in Fig.~\ref{ROC_example}. The ROC of the FC consists of eight (i.e., $2^3$) segments with nonincreasing slopes. The two thin solid magenta segments extend the first and last segments. The ROC of sensor $i$'s local detection comprises two (i.e., $2^1$) segments which are split by $(q_i,p_i)$, denoted as a hollow circle. Moreover, $R_1\geq R_2 R_3$, hence by Lemma~\ref{better_performance_2}, 
we have $p_0=p_1$ when $q_0=q_1$. Therefore, the curve of the FC is always above 
that of sensor $1$ except for $(0.16,0.74)$ at which they coincide.
\end{example}

\begin{figure}[htbp]
	\centering
	\includegraphics[width=0.48\textwidth]{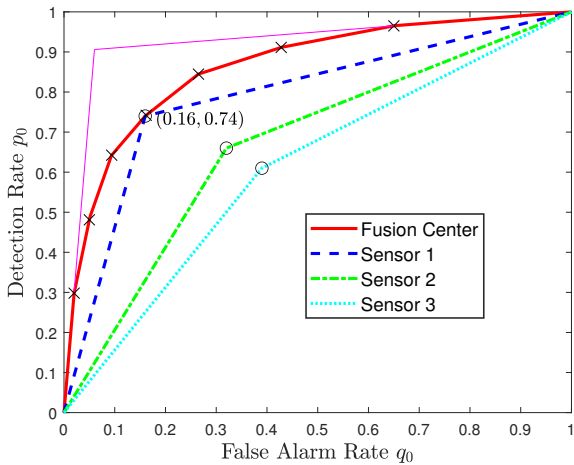}
	\caption{The ROCs of the four detectors in Example~\ref{example:ROC_four_detectors}. }
	\label{ROC_example}
\end{figure}

\section{The Oracle Optimal Algorithm} \label{oracle_optimal_algorithm}
In this section, we will discuss the optimal algorithm that exactly solve~\eqref{N-P_test_original}. We call it the oracle optimal algorithm as its detection probability at every stage is not available analytically and explicitly. We will study its asymptotic performance. 

By~\eqref{LR_test_0_one_stage}, the optimal fusion rules of~\eqref{N-P_test_original} at the $k$-th stage is in terms of the following likelihood ratio (LR) tests:
\begin{equation} \label{LR_test_1}
\Lambda^k: = \frac{\prod_{i=1}^{n} P(u_i^k|H_1)P(u_0^{k-1}|H_1)}{\prod_{i=1}^{n} P(u_i^k|H_0)P(u_0^{k-1}|H_0)}\mathop {\begin{array}{*{20}{c}}
	\ggll
	\end{array}}\limits_{{H_0}}^{{H_1}} (t^k,\lambda^k). 
\end{equation} 
where $t^k$ and $\lambda^k$ are the threshold and the random factor at time $k$.  In particular, when $k=1$, we let $P(u_0^0 \mid H_1)=P(u_0^0 \mid H_0)=0.5$ to make the LR test fit the form of~\eqref{LR_test_1}. To solve~\eqref{LR_test_1}, similar to~\eqref{N-P_test_onestage}, we can resort to solving the following problem: 
\begin{equation} \label{N-P_test}
\begin{aligned} 
\mathop{\rm maximize}_{t^k\in \mathbb R_+,\lambda^k\in [0,1]} \quad & p_0^k(t^k, \lambda^k) \\
{\rm subject~to} \quad & q_0^k(t^k, \lambda^k) \le \alpha.
\end{aligned}
\end{equation}

In virtue of the result in Section~\ref{preliminaries}, we can calculate the optimal pair of $(t^k,\lambda^k)$ as follows.  The likelihood ratio $\Lambda^k$ has $2^{n+1}$ possible values. Let ${\bf \Lambda}^k$ denote the assemble of the possible values of $\Lambda^k$.
Sorting
$\Lambda_1^k,\ldots, \Lambda_{2^{n+1}}^k$ in a non-decreasing order and finding $\Lambda_*^k$ similar to~\eqref{find_optimal_threshold}, then optimal $t^k$ and  $\lambda^k$ are computed as follows:
\begin{equation} \label{oracle_threshold_multistage}
\begin{cases}
t^k=\Lambda_*^k, \\
\lambda^k=\frac{\alpha-\sum_{\Lambda^k \in {\bf \Lambda}^k: \Lambda^k>\Lambda_*^k} P(\Lambda^k|H_0)}{P(\Lambda_*^k|H_0)}.
\end{cases}       
\end{equation}

We will study the asymptotic performance of the oracle optimal algorithm. In the fusion algorithm~\eqref{LR_test_1}, the one-bit memory can encode historical local decisions from the sensors in a low-resolution format and include
it in the iteration of global decision making.  The next result is on the asymptotic behavior of the fusion detection probability, which unveils improved detection performance due to the memory setup for the oracle algorithm. 
\begin{theorem} \label{oracle_performance}
The detection probability of the FC using policy~\eqref{LR_test_1} converges, i.e., 
\begin{equation}\label{asymptotic_detection_probability}
 p_0^{\infty}:=\lim\limits_{k \to \infty} p_0^k 
\end{equation} exists, and for $n \geq 2$ sensors, 
	\begin{equation}\label{oracle_symptotic_performance2}
	p_0^{\infty}=\frac{\alpha \left(\prod_{i=1}^{n}R_i\right)}{1+\alpha \left(\prod_{i=1}^{n}R_i-1\right)},
	\end{equation}
   where $R_i$ is defined in~\eqref{new_performance_metric}.
\end{theorem}

The proof of Theorem~\ref{oracle_performance} is presented in Appendix~\ref{appendix_oracle_performance} which consists of two parts. In the first part, we show that the sequence $(R_0^k)_{k=1}^{\infty}$ is monotonically increasing and upper bounded by $\prod_{i=1}^{n} R_i$ which implies $R_0^{\infty}:=\lim\limits_{k \rightarrow \infty} R_0^k$ is well-defined, where 
\begin{equation*}%\label{new_performance_metric2}
R_0^k=\frac{p_0^k}{1-p_0^k}\frac{1-q_0^k}{q_0^k}.
\end{equation*}In the second part, we construct a fusion algorithm for which $R_0^k$ converges to $\prod_{i=1}^{n} R_i$. Then~\eqref{oracle_symptotic_performance2} follows by squeeze principle.

We note that $R_0^1$ depicts the performance of the FC without memory. Since the sequence $(R_0^k)_{k=1}^{\infty}$ is monotonically increasing, the one-bit memory can indeed improve the detection performance. However, $R_0^1$ is not analytically available, and we cannot quantify the value of the one-bit memory. Nevertheless, from the simulation and real-world experiment results, it can be seen that the improvement is significant. 

\begin{remark} \label{asymptotic_comparison}
In~\cite{kam1999performance}, under the Bayesian formulation, the result unveils that the false alarm and detection probabilities of the FC will eventually enter a certain polygon in the $p-q$ plane. Once the probabilities
have entered the region, they remain unchanged, and no further
observations need to be collected. In our problem setting, it can be read from Theorem~\ref{oracle_performance} that there also exists a ``stopping point'' $p_0^{\infty}$ of detection probability. However, $p_0^k$ cannot exceed $p^{\infty}_0$, but infinitely approach it.   
\end{remark}

Following from on Theorem~\ref{oracle_performance}, the following corollary is straightforward:
\begin{corollary} \label{corollary_convergence}
	The threshold and random factor $(t^k,\lambda^k)$ in~\eqref{LR_test_1} converge, i.e., the following limit values exist:
		\begin{equation}\label{threshold_convergence}
	t^{\infty}:=\lim\limits_{k \to \infty} t^k \hbox{~~~~~and~~~~~~} \lambda^{\infty}:=\lim\limits_{k \to \infty} \lambda^k. 
	\end{equation}
\end{corollary}

\begin{remark}
	We remark that although the sequence of $(t^k,\lambda^k)_{k=1}^{\infty}$ can be pre-determined, this offline manner needs infinite memory to store these parameters. Hence, we suppose that the parameters are calculated using an online manner. At each stage, $2^{n+1}$ likelihood ratios need to be calculated and sorted to solve the N-P test. This procedure is computationally inefficient, and it cannot be scaled when $n$ goes large.
\end{remark}

\section{The Low-complexity Distributed Detection with One-bit Memory} \label{the_proposed_algorithm}

We see from the last section that the oracle detection algorithm is computationally inefficient in the sense that a test is performed at each time subject to the N-P criterion. In this section, we will turn to investigate an algorithm with low computational complexity.  To this end,  a part of performance needs to be sacrificed. Here it will be the transient one. The reason why we can pay less attention to transient performance is that in the proposed low-complexity algorithm, as we will show, its detection probability converges to a steady state exponentially at the same rate as that of the oracle algorithm (Theorem~\ref{convergence_sped}). Taking a long-term point of view, we may not lose much when we only concentrate on the steady-state detection probability, and alternatively, we consider the following problem:
\begin{equation} \label{proposed_N-P_test_my}
\begin{aligned} 
\mathop{\rm{maximize}}_{\bm{\Gamma}} \quad & \lim \mathop{\rm{inf}}_{k \rightarrow \infty} p_0^k(\bm{\Gamma})\\
{\rm subject~to} \quad & q_0^k(\bm{\Gamma}) \leq \alpha,~k=1,2,\ldots.
\end{aligned}
\end{equation}
Before solving the problem~\eqref{proposed_N-P_test_my}, we will first focus on the following problem with relaxed constraint: 
\begin{equation} \label{proposed_N-P_test2}
\begin{aligned} 
\mathop{\rm{maximize}}_{\bm{\Gamma}} \quad & \lim \mathop{\rm{inf}}_{k \rightarrow \infty} p_0^k(\bm{\Gamma})\\
{\rm subject~to} \quad & \lim \mathop{\rm{inf}}_{k \rightarrow \infty} q_0^k(\bm{\Gamma}) \leq \alpha.
\end{aligned}
\end{equation}
In what follows, we will devise a fusion policy that can be seen as a stationary approximation of the oracle optimal hypothesis test. Then we will reveal that as time goes the two algorithms will eventually approach the same detection performance.  
The oracle and the proposed algorithms will be further compared from the perspectives of computational complexity as well as convergence speed. 

\subsection{Algorithm Development and Analysis}
We propose a stationary fusion algorithm,  i.e., in which $\Gamma^k$ will not vary in different stages, which thereby has a low computational complexity.  In the algorithm, at each time $k$ the following LR test is performed:
\begin{equation} \label{LR_test_5}
	\frac{\prod_{i=1}^{n} P(u_i^k|H_1)}{\prod_{i=1}^{n} P(u_i^k|H_0)}\mathop {\begin{array}{*{20}{c}}
			\ggll
	\end{array}}\limits_{{H_0}}^{{H_1}} \left(t(u_0^{k-1}),\lambda(u_0^{k-1}) \right) , 
\end{equation} 
where  the threshold and random factor are  selected according to the global detection result of the previous time, formally expressed as follows:
\begin{equation} \label{threshold_selection}
\left(t(u),\lambda(u) \right)=
\begin{cases}
(t_0,\lambda_0), &\text{if } u=0, \\
(t_1,\lambda_1), &\text{if } u=1,
\end{cases}       
\end{equation}
where $(t_0,\lambda_0)$ and $(t_1,\lambda_1)$ are to be determined later.  Since  $u_0^{0}$ is not available  initially at time $k=1$,  $\left(t(u_0^0),\lambda(u_0^0) \right)$ can be selected as either $(t_0,\lambda_0)$ or $(t_1,\lambda_1)$.
\begin{remark} \label{relationship_of_algorithms}
The devising of detection rule~\eqref{LR_test_5} is enlightened from the convergence of the oracle test rules. From Theorem~\ref{oracle_performance} and Corollary~\ref{corollary_convergence}, we have the stationary values $p_0^{\infty}$, $t^{\infty}$ and $\lambda^{\infty}$ when the detection probability of the FC converges. At that time~\eqref{LR_test_1} becomes 
\begin{equation*} \label{LR_test_4}
	 \frac{\prod_{i=1}^{n} P(u_i^k|H_1)}{\prod_{i=1}^{n} P(u_i^k|H_0)}\mathop {\begin{array}{*{20}{c}}
		\ggll
		\end{array}}\limits_{{H_0}}^{{H_1}} \left(\frac{P(u_0^{k-1}|H_0)}{P(u_0^{k-1}|H_1)} t^{\infty},\lambda^{\infty} \right),
	\end{equation*}
	where
	\begin{equation*} \label{threshold_selection_2}
	\left(\frac{P(u_0^{k-1}|H_0)}{P(u_0^{k-1}|H_1)} t^{\infty},\lambda^{\infty} \right)=
	\begin{cases}
	(\frac{1-\alpha}{1-p_0^{\infty}} t^{\infty},\lambda^{\infty}), &\text{if } u_0^{k-1}=0, \\
	(\frac{\alpha}{p_0^{\infty}} t^{\infty},\lambda^{\infty}), &\text{if } u_0^{k-1}=1.
	\end{cases}       
	\end{equation*}
From this perspective, the proposed algorithm~\eqref{LR_test_5} and~\eqref{threshold_selection} can be treated as an approximate fusion algorithm of the oracle policy. 
    \end{remark}

We will analyze the proposed algorithm~\eqref{LR_test_5}. The following lemma shows the convergence of the detection and false alarm probabilities of the FC equipped with~\eqref{LR_test_5} and~\eqref{threshold_selection}. It was firstly presented in~\cite{zeng2018novel}, 
which is rephrased to adapt itself to the context of the paper.

Let $p_{0,0}$ and $q_{0,0}$ denote the detection probability and false alarm probability of the FC at time $k=1$ when  $\left(t(u_0^0),\lambda(u_0^0) \right)=(t_0,\lambda_0)$, and $p_{0,1}$ and $q_{0,1}$ denote the detection probability and false alarm probability of the FC at time $k=1$ with $\left(t(u_0^0),\lambda(u_0^0) \right)=(t_1,\lambda_1)$.

\begin{lemma}[Theorem~1~\cite{zeng2018novel}]\label{alalgorithm_convergence}
	Given Assumptions~\ref{unchanged_true_hypothesis}, by using the proposed algorithm~\eqref{LR_test_5} and~\eqref{threshold_selection}, the detection and false alarm probabilities of the FC converge, i.e., 
	\begin{equation}\label{performance_convergence}
	p_0^{\infty}:=\lim\limits_{k \to \infty} p_0^k \hbox{~~and~~} q_0^{\infty}:=\lim\limits_{k \to \infty} q_0^k 
	\end{equation}
	exist, and the limits are
	\begin{equation}\label{convergence_performance}
	p_0^{\infty}= \frac{p_{0,0}}{1-(p_{0,1}-p_{0,0})}
	\end{equation}
and
\begin{equation} \label{convergence_of_falsealarm_probability}
	q_0^{\infty}= \frac{q_{0,0}}{1-(q_{0,1}-q_{0,0})}.  
	\end{equation}
\end{lemma}

Combining~\eqref{LR_test_5},~\eqref{threshold_selection} and Lemma~\ref{alalgorithm_convergence} together, we conclude that~\eqref{proposed_N-P_test2} is equivalent to
choosing $(t_0,\lambda_0)$ and $(t_1,\lambda_1)$ so that $p_0^{\infty}$ is maximized subject to 
$q_0^{\infty}\leq \alpha$. It is worth noting that similarly in~\eqref{N-P_test_onestage} $p_0^{\infty}$ reaches its maximum when $q_0^{\infty}=\alpha$. From~\eqref{convergence_of_falsealarm_probability} we have $q_{0,1}=1+q_{0,0}-\frac{q_{0,0}}{\alpha}$ when $q_0^{\infty}=\alpha$. Moreover, $p_{0,0}$ is determined once $q_{0,0}$ is fixed and so do $p_{0,1}$ and $q_{0,1}$. Therefore, to solve~\eqref{proposed_N-P_test2} it relies on solving the following problem:
\begin{equation} \label{proposed_N-P_test}
\begin{aligned} 
\mathop{\rm maximize}_{q_{0,0},q_{0,1} \in [0,1]} \quad & p_0^{\infty}(q_{0,0},q_{0,1}) \\
{\rm subject~to} \quad & q_{0,1}=1+q_{0,0}-\frac{q_{0,0}}{\alpha},
\end{aligned}
\end{equation}
Once $q_{0,0}$ and $q_{0,1}$ are determined, the optimal $(t_0,\lambda_0)$ and $(t_1,\lambda_1)$ can be uniquely obtained. 

 We will explore the optimal $q_{0,0}$ in~\eqref{proposed_N-P_test} by exploiting the segmentation characteristic of the ROC curve of the optimal single-stage N-P decision fusion, see Lemma~\ref{ROC_curve_one_stage}. The following theorem formally presents the conclusion.

\begin{theorem} \label{possible_t1}
	An optimal solution to problem~\eqref{proposed_N-P_test} can be expressed as $(q_{0,0},q_{0,1})=(q,1+q-q/\alpha)$, where $q\in \left( 0,q'^*_1\right] $ with $q'^*_1=\min \left\{\prod_{i=1}^{n} q_i,\frac{\alpha}{1-\alpha} \prod_{i=1}^{n} (1-q_i)\right\}$. There corresponds to an optimal value of 
    $p_0^\infty$, which is
		\begin{equation}\label{oracle_symptotic_performance3}
	p_0^{\infty}=\frac{\alpha \left(\prod_{i=1}^{n}R_i\right)}{1+\alpha \left(\prod_{i=1}^{n}R_i-1\right)}.
	\end{equation}
\end{theorem}
\begin{proof}
The proof utilizes the segmented property of the ROC curve of one-stage decision fusion mentioned in Section~\ref{preliminaries}. First we will prove that when $q_{0,0} \in \left( 0,q'^*_1\right] $, the equation~\eqref{oracle_symptotic_performance3} holds. When $q_{0,0} \in \left( 0,q'^*_1\right] $, i.e., $(q_{0,0},p_{0,0})$ locates on the leftmost segment of the ROC curve, we obtain $t_0=\Lambda_{2^n}$ and $p_{0,0}=l_1 q_{0,0}$. 
For $t_1$, note that $q_{0,1} =1+q_{0,0}-\frac{q_{0,0}}{\alpha}\geq 1-\prod_{i=1}^{n}(1-q_i)=\sum_{i=2}^{2^n} P(\Lambda_i|H_0)$, hence $(q_{0,1},p_{0,1})$ locates on the rightmost segment and we obtain $t_1=\Lambda_1$ and $p_{0,1}=l_{2^n}q_{0,1}+(1-l_{2^n})=1-l_{2^n}\frac{1-\alpha}{\alpha}q_{0,0}$. Note that $l_1=\frac{\prod_{i=1}^{n}p_i}{\prod_{i=1}^{n}q_i}$ and $l_{2^n}=\frac{\prod_{i=1}^{n}(1-p_i)}{\prod_{i=1}^{n}(1-q_i)}$, we have 
\begin{equation} \label{optimal_case}
p_0^{\infty}=\frac{p_{0,0}}{1-(p_{0,1}-p_{0,0})}=\frac{\alpha \left(\prod_{i=1}^{n}R_i\right)}{1+\alpha \left(\prod_{i=1}^{n}R_i-1\right)}.
\end{equation}

Next, we will prove that $q_{0,0} \in \left( 0,q'^*_1\right] $ is optimal. Since $l_1>\dots>l_{2^n}$, we have 
$p_{0,0} \leq l_1 q_{0,0}$ and $p_{0,1} \leq 1+l_{2^n}(q_{0,1}-1)=1-l_{2^n}\frac{1-\alpha}{\alpha}q_{0,0}$. When $q'^*_1<q_{0,0}$, it can be verified that the \enquote{$=$} cannot simultaneously hold for both inequalities. (Again take the case of three sensors in Fig.~\ref{ROC_example} as example. We can see that the $2$-th segment to the $7$-th segment of the ROC curve are below the two thinner solid magenta segments. When $q'^*_1<q_{0,0}$, at least one of $(q_{0,0},p_{0,0})$ and $(q_{0,1},p_{0,1})$ locates in the interval from the $2$-th segment to the $7$-th segment of the ROC curve). From the first part of the proof we see that when $(q_{0,0},p_{0,0})$ and $(q_{0,1},p_{0,1})$ locate on solid magenta segments respectively, we have~\eqref{optimal_case}. For $(q_{0,0},p_{0,0})$ and $(q_{0,1},p_{0,1})$ under the two solid magenta segments, it can be verified that 
\begin{align*}
p_0^{\infty} =\frac{p_{0,0}}{1-(p_{0,1}-p_{0,0})} &<\frac{l_1q_{0,0}}{1-\left(1-l_{2^n} \frac{1-\alpha}{\alpha}q_{0,0}-l_1 q_{0,0} \right)} \\
&=\frac{\alpha \left(\prod_{i=1}^{n}R_i\right)}{1+\alpha \left(\prod_{i=1}^{n}R_i-1\right)},
\end{align*}
which completes the proof.
\end{proof}

We know form the proof of Theorem~\ref{possible_t1} that in the optimal case, $(q_{0,0},p_{0,0})$ locates on the leftmost segment of the ROC curve and $(q_{0,1},p_{0,1})$ locates on the rightmost segment, hence we have that $q_{0,0}<q_{0,1}$. Such a relation can also be inferred from intuition. If $u_0^{k-1}= 0$ which means $H_0$ is more
	likely the true hypothesis than $H_1$, it is reasonable to increase the threshold to reduce the false alarm probability, i.e.,
	choose $(t_0,\lambda_0)$ with a small false alarm probability, and vice versa. We also know that to obtain the optimal $q_{0,0}$ and $q_{0,1}$ which solves~\eqref{proposed_N-P_test}, we can let 
	\begin{equation} \label{threshold_selection_3}
	\begin{cases}
	(t_0,\lambda_0)=\left(\Lambda_{2^n},\frac{q'^*_1}{\prod_{i=1}^{n} q_i}\right), \\
	(t_1,\lambda_1)=\left(\Lambda_{1},1+\frac{(\alpha-1)q'^*_1}{\alpha\prod_{i=1}^{n} (1-q_i)}\right),
	\end{cases}    
	\end{equation}
	where $q'^*_1=\min \left\{\prod_{i=1}^{n} q_i,\frac{\alpha}{1-\alpha} \prod_{i=1}^{n} (1-q_i)\right\}$.
	It is noteworthy that we can compute the above two pair of parameters before system runs according to the configuration of the sensors, which are know \emph{a priori}. 

Now, we are going to tackle problem~\eqref{proposed_N-P_test_my} by a little more effort. From Lemma~\ref{alalgorithm_convergence} we know that the steady-state performances do not depend on the initial threshold that we choose from the two alternatives. However, the initial threshold can affect the transient performances. If we set $(t^1,\lambda^1)$ as $(t_0,\lambda_0)$ in~\eqref{threshold_selection_3}, we can obtain the following result by induction~\cite{zeng2018novel}: 
\begin{align*}
q_0^k &=q_{0,0} \frac{1-(q_{0,1}-q_{0,0})^k}{1-(q_{0,1}-q_{0,0})}\\
&<q_0^{\infty}=\alpha,
\end{align*}
which implies that the transient false alarm probabilities can always be smaller than the desired probability $\alpha$. This will be also illustrated in simulations, see Fig.~\ref{convergence_of_falsealarm}. Hence, the fusion algorithm~\eqref{LR_test_5} and~\eqref{threshold_selection} with initial threshold $(t_0,\lambda_0)$ in~\eqref{threshold_selection_3} optimally solve the problem~\eqref{proposed_N-P_test_my}.

 \subsection{Comparison with Oracle Optimal Algorithm: Computational Complexity and Convergence Rate}   
 
By comparing~\eqref{oracle_symptotic_performance2} and~\eqref{oracle_symptotic_performance3}, we know that both algorithms have the same asymptotic detection probability. Now we consider their computational complexity and convergence speed. We compare the computational complexity by considering their time complexity at a single stage. For the oracle optimal algorithm, at each stage, $2^{n+1}$ possible likelihood ratios should be sorted in an increasing order, for which the computation complexity is $O(2^{2n+2})$ in the worst case by means of some selection sort algorithms. Moreover, appropriate threshold and random factor need to be selected according to to~\eqref{oracle_threshold_multistage}, requiring $O(2^{n+1})$ comparisons in the worst case. Additionally, $n+1$ multiplications are required to calculate the likelihood ratio given $n$ local messages and one-bit memory. Overall, the computational complexity of the oracle optimal algorithm is $O(4^n)$, which is less efficient, especially when $n$ is large. In contrast, for threshold generation in our proposed algorithm, at each stage, the FC only needs to adjust its threshold and random factor from a set of two candidates according to its global detection decision in the last stage, which is stored in its memory. There corresponds to a constant time complexity $O(1)$ despite the size $n$ of the sensor network. By adding the computation of likelihood ratio, the computational complexity of the proposed algorithm is $O(n)$, which shows great superiority to the oracle one.

We see that both algorithms will converge to the same steady-state detection probability. In terms of convergence rate, in the following theorem, we show that their convergence is exponential subject to an asymptotically identical rate.
\begin{theorem} \label{convergence_sped}
	For both algorithms~\eqref{LR_test_1} and~\eqref{LR_test_5}, we have  
	\begin{equation}
	\lim\limits_{k \rightarrow \infty} \frac{\left| p_0^{k+1}-p_0^{\infty}\right| }{\left| p_0^{k}-p_0^{\infty}\right| }=a \in (0,1),
	\end{equation}
	where 
	\begin{equation*}
	a=
	\begin{cases}
	1-\left(1+Q_3 \right) \prod_{i=1}^{n} p_i,~\text{if} ~1<Q_4, \\
		1-\left(1+\frac{1}{Q_3} \right) \prod_{i=1}^{n} (1-p_i),~\text{if} ~1 \geq Q_4.
	\end{cases}
	\end{equation*}
		where $Q_3=\frac{1-\alpha}{\alpha \prod_{i=1}^{n} R_i}$, and $Q_4=\frac{\alpha \prod_{i=1}^{n} (1-q_i)}{(1-\alpha)\prod_{i=1}^{n} q_i}$.
\end{theorem}
The proof of Theorem~\ref{convergence_sped} is presented in Appendix~\ref{appendix_convergence_sped}.

\section{Numerical Simulations and Real-world Experiments} \label{simulation_experiment_results}
\subsection{Numerical Simulations}
We simply consider the detection fusion of $n$ homogeneous sensors in the sense that they have identical observation model and local decision rule. In the text below, the index of the sensors will be dropped from notations when it does not cause ambiguity. The observation model of every sensor is as follows~\cite{zhu2010fusion}:
\begin{equation} \label{observation_model}
Y^k=
\begin{cases}
w^k, &H_0,\\
A+w^k, &H_1,
\end{cases}       
\end{equation}
where $Y^k$ is the sensor measurement at the $k$-th stage, $A$ is a constant, and $w^k$ is a zero-mean Gaussian noise with variance $\sigma^2$. The local decision rule is set as the threshold rule:
\begin{equation} \label{decison _rule_of_sensors}
u^k=
\begin{cases}
1, &\text{if } y^k \geq y^*,\\
0, &\text{if } y^k <y^*,
\end{cases}       
\end{equation} 
where $y^k$ is the realization of $Y^k$ at the $k$-th stage and $y^*$ is the decision threshold. Since the outputs of local sensors are binary, $p$ ($p_i$) and $q$ ($q_i$) in~\eqref{local_performance} can be viewed as local detection and false alarm probabilities. From~\eqref{observation_model} and~\eqref{decison _rule_of_sensors}, we obtain:
\begin{equation} \label{calculation_detection}
p=\int_{y^*}^{\infty} P(y^k|H_1) d y^k  \hbox {~~and~~}
q=\int_{y^*}^{\infty} P(y^k|H_0) d y^k.
\end{equation}
Note that we do not focus on the decision rules at local sensors. With no loss of generality, we set $A=2$ and $y^*=1$ throughout our simulations. We let $\alpha=q$ to conveniently compare the detection performance between the sensors and the FC.

In the first trial, we set $\sigma ^2=5$ and the signal to noise ratio\footnote{SNR is calculated as SNR=$10\log(A/\sigma^2)$.} (SNR) at local sensors is $-4$dB.
Based on~\eqref{calculation_detection} we obtain $p=0.67$ and $q=0.33$. 
We perform different fusion mechanisms and compare their performance in terms of detection probability. We consider the oracle N-P test~\eqref{LR_test_1} and the proposed low-complexity algorithm~\eqref{LR_test_5}, which can be categorized as distributed detection using a one-bit memory.   Another fusion detector we consider is a kind of instant fusion detection, where the FC only fuses the latest local decisions from the sensors for global decision making. When this detector runs, it implements~\eqref{LR_test_0_one_stage} at every time stage. All the considered fusion detectors have the same asymptotic false alarm probability as the sensor does. How the (steady-state) detection probability of these detectors changes with respect to different numbers of sensors is demonstrated in Fig.~\ref{detection_probability_sensor_number}.  The asymptotic detection probabilities of the oracle test and the proposed algorithm are calculated based on~\eqref{oracle_symptotic_performance2} and~\eqref{oracle_symptotic_performance3}. The detection probabilities of the FC without memory are calculated based on~\eqref{detection_probability_calculation} where the threshold and randomization factor are obtained by~\eqref{oracle_threshold}. We see from the figure that the detection probabilities of the fusion detectors~\eqref{LR_test_1},~\eqref{LR_test_5} are identical and that as sensor number increases, the detection probabilities of all fusion detectors approach 1. The detectors~\eqref{LR_test_1},~\eqref{LR_test_5}  perform much better than~\eqref{LR_test_0_one_stage}  when sensors are few. It is noteworthy that the detectors~\eqref{LR_test_1},~\eqref{LR_test_5} have the same detection probability as a single sensor when $n=1$, while so the detector~\eqref{LR_test_0_one_stage} does when $n \in\{1,2\}$. This happens because the detection probability of the FC (in the canonical N-P setup) can exceed that of each sensor (under the constraint that they have the same false alarm probabilities) only when $n\geq 3$, see~\cite{xiang2001performance} and the one-bit memory storage can be viewed as an additional node, hence for $n=2$ we actually fuse three local decisions at each stage.

\begin{figure}[htbp]
	\centering
	\includegraphics[width=0.48\textwidth]{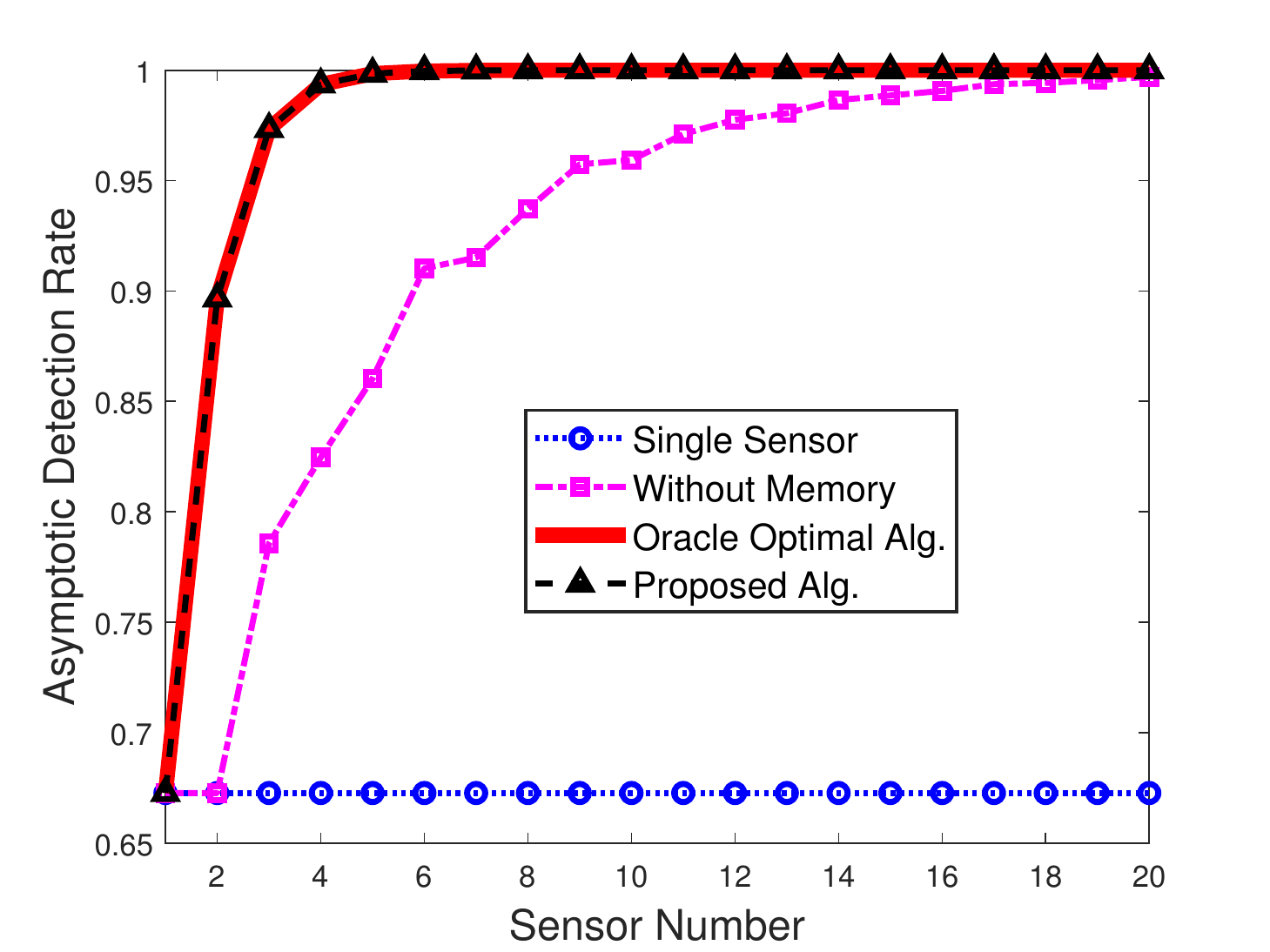}
	\caption{The relation between the (steady-state) detection probabilities and the number of sensors for different detection algorithms in simulations.}
	\label{detection_probability_sensor_number}
\end{figure}

Next, we set $n=4$ and let the SNR at local sensors be at  $10$ different levels, ranging from$-10$dB to $8$dB, respectively, to illustrate the relation between the (steady-state) detection probabilities of various detectors and the SNRs, see Fig.~\ref{detection_probability_snr}. We see from the figure that the fusion detectors~\eqref{LR_test_1},~\eqref{LR_test_5} have greatly improved detection performance, especially when the SNR is low. Moreover, the steady-state detection probability of our proposed algorithm resembles that of the oracle one, validating our theoretical claims. 

\begin{figure}[htbp]
	\centering
	\includegraphics[width=0.48\textwidth]{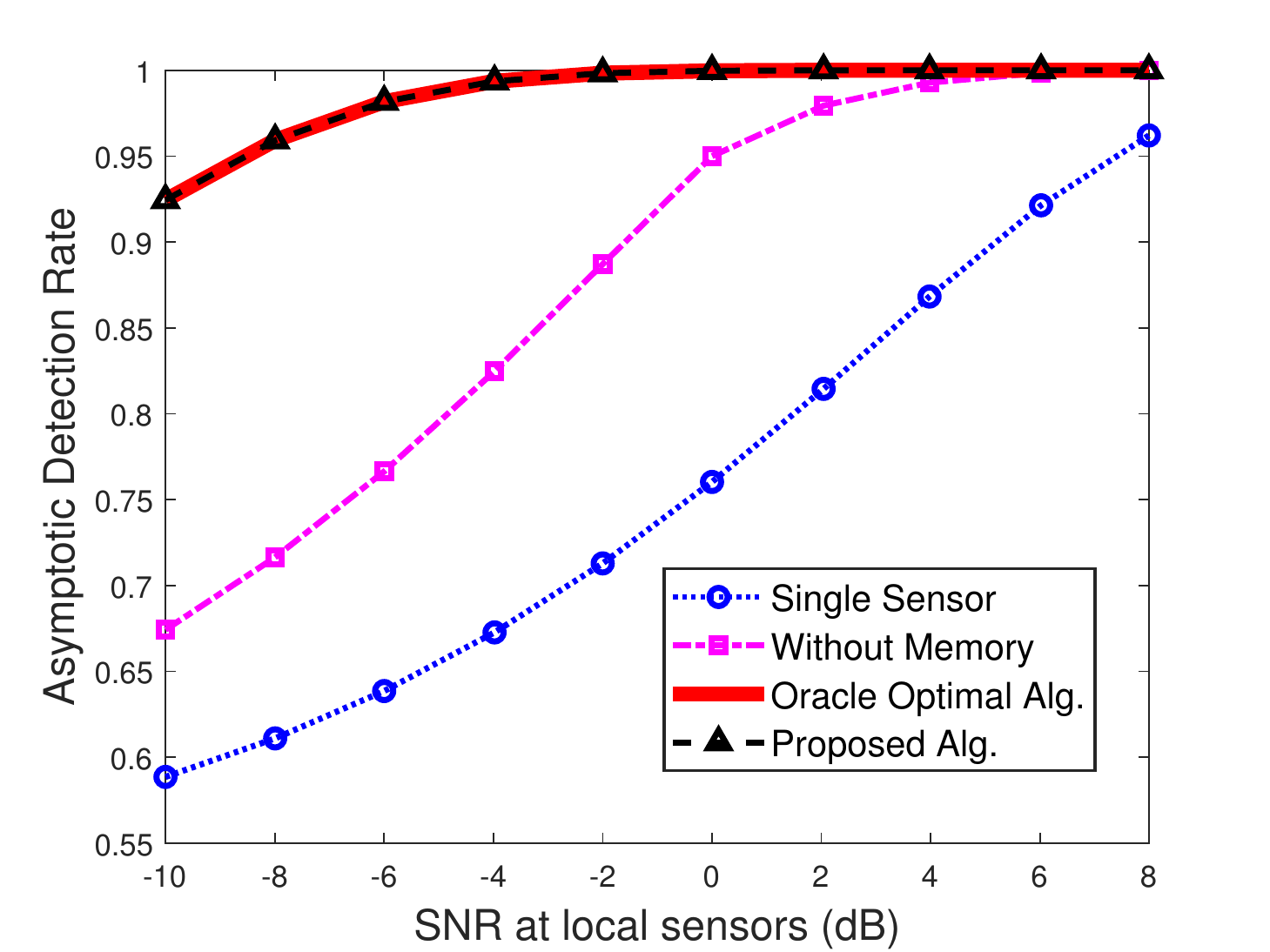}
	\caption{The relation between the (steady-state) detection probabilities and SNR of the local (homogeneous) sensor for different detection algorithms in simulation.}
	\label{detection_probability_snr}
\end{figure}

Let SNR be -8dB, by~\eqref{calculation_detection} we obtain $p=0.61$ and $q=0.39$. The false alarm and the detection probabilities over time are respectively plotted in Fig.~\ref{convergence_exhibition} for different detectors. It is notable that for each plot, our proposed algorithm is simulated twice for two different initial thresholds. We see that they will converge to the same value exponentially despite different threshold initialization. For the oracle optimal algorithm, its detection probability seems to converge faster when $k$ is small. In addition, it maintains a constant false alarm probability over time while our algorithm does not, as Fig.~\ref{convergence_of_falsealarm} exhibits. 
Nevertheless, if we set the initial threshold in~\eqref{LR_test_5} as $(t^1,\lambda^1)=(t_0,\lambda_0)$, where $(t_0,\lambda_0)$ can be calculated from~\eqref{threshold_selection_3}, we can guarantee that the resulting false alarm probability is smaller than $\alpha$ exactly at each time. We also note from Fig.~\ref{convergence_exhibition} that the proposed algorithm needs dozens of time stages to achieve the stationary state. Actually, if the time interval between two successive decisions is in seconds, which is often the case in target detection systems~\cite{shi2018anti}, then tens of seconds is enough for convergence. Hence the assumption of the true hypothesis remaining unchanged is realistic to some extent. 

\begin{figure*}[htbp]
	\centering
	\begin{subfigure}{0.96\columnwidth}
		\includegraphics[width=\columnwidth]{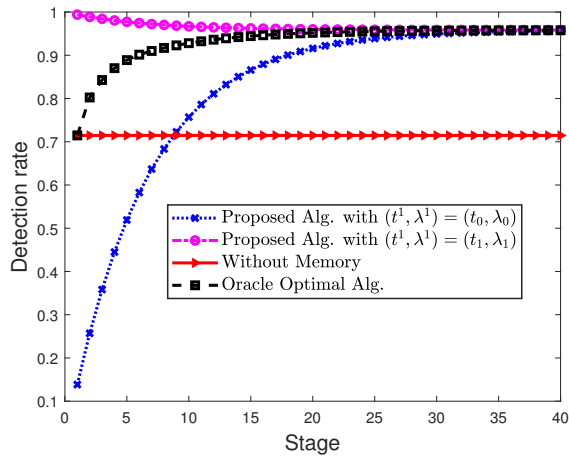}%
		\caption{Convergence of detection probability}
		\label{convergence_of_detection}
	\end{subfigure}
	\hspace{0.02\textwidth}
	\begin{subfigure}{.96\columnwidth}
		\includegraphics[width=\columnwidth]{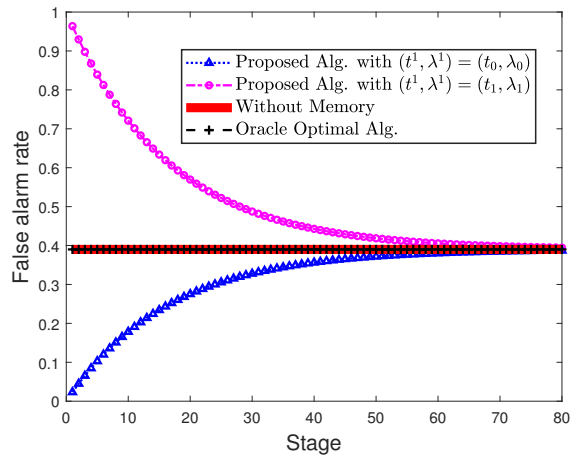}%
		\caption{Convergence of false alarm probability}
		\label{convergence_of_falsealarm}
	\end{subfigure}\hfill%
	\caption{Convergence of the detection performance probabilities. For each plot, our proposed algorithm is simulated twice: one that corresponds to the blue line is the case where the initial threshold is set to be $(t_0,\lambda_0)$,  the other that corresponds to the magenta line is the case where the initial threshold is set to be $(t_1,\lambda_1)$.}
	\label{convergence_exhibition}
\end{figure*}

It should be noted that till now the optimal threshold $q_{0,0} \in (0,q'^*_1]$ is adopted for our proposed algorithm. Based on Theorem~\ref{possible_t1}, it optimally solves problem~\eqref{proposed_N-P_test} and has an identical asymptotic detection probability with that of the oracle optimal algorithm.
Next, we will plot $p_0^{\infty}$ against $q_{0,0}$. For the case of four homogeneous sensors with $p=0.61$ and $q=0.39$, we have $q'^*_1=(q)^4$ and the simulation result is shown in Fig.~\ref{optimal_t1}. The solid segment denotes the case of $0<q_{0,0} \leq (q)^4$, while the dashed segments represent the case of $(q)^4 <q_{0,0} \leq q$. We can see from the figure that the asymptotic detection probability $p_0^{\infty}$ remains unchanged when $0<q_{0,0} \leq (q)^4$, and it decreases with the increase of $q_{0,0}$ when $(q)^4<q_{0,0} \leq q$, which coincides with Theorem~\ref{possible_t1}. We notice that the curve is also segmented which is due to piecewiseness of the ROC curve. However, the number of pieces may vary in different situations. 

\begin{figure}[htbp]
	\centering
	\includegraphics[width=0.46\textwidth]{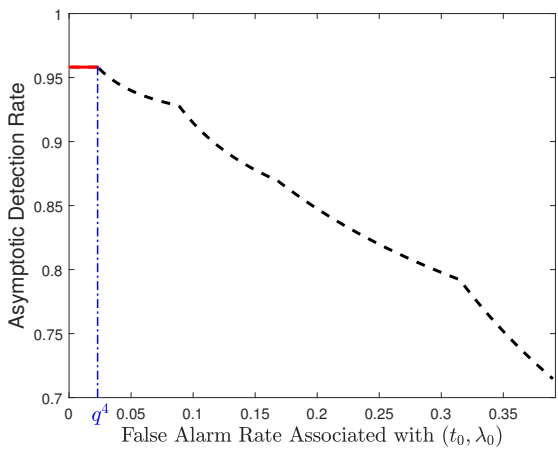}
	\caption{The relation between the (steady-state) detection probability of the proposed algorithm and the false alarm probability $q_{0,0}$ associated with $(t_0,\lambda_0)$. The solid curve denotes the case of $0<q_{0,0} \leq (q)^4$ and the dashed curve denotes the case of $(q)^4<q_{0,0} \leq q$.}
	\label{optimal_t1}
\end{figure}

\subsection{Real-world Experiments}
In this part, we present the experimental results of acoustic source detection using a microphone array. The deployment of the microphone array is shown in Fig.~\ref{microphone_array}, and the device parameters are summarized in TABLE~\ref{microphone_characteristics}. The acoustic source is a DJI Phantom 2 drone, see Fig.~\ref{drone_type}.

\begin{table}[htbp]
	\centering
	\captionsetup{justification=centering}
	\caption{Characteristics of the microphones} \label{microphone_characteristics}
		\resizebox{\columnwidth}{!}{%
	\begin{tabular}{|c |c|}
		\hline
		Model & PCB 130A24\\ [0.5ex]
		\hline
		Frequency Response ($\pm$3 dB)& 20-16000 Hz\\
		\hline
		Sound Field & Free-Field\\
		\hline
		Sensitivity (@ 250 Hz) & 10 mV/Pa\\
		\hline
		Inherent Noise (A Weighted)& $<$30 dB(A) re 20$\mu$Pa\\
		\hline
		Dynamic Range (3$\%$ Distortion Limit) & $>$143 dB re 20$\mu$Pa\\
		\hline
	\end{tabular}
}
\end{table}

\begin{figure}[htbp]
	\captionsetup{singlelinecheck = false, justification=justified}
	\captionsetup[subfigure]{justification=centering}
	\begin{subfigure}[t]{0.57\columnwidth}
		\includegraphics[width=\linewidth]{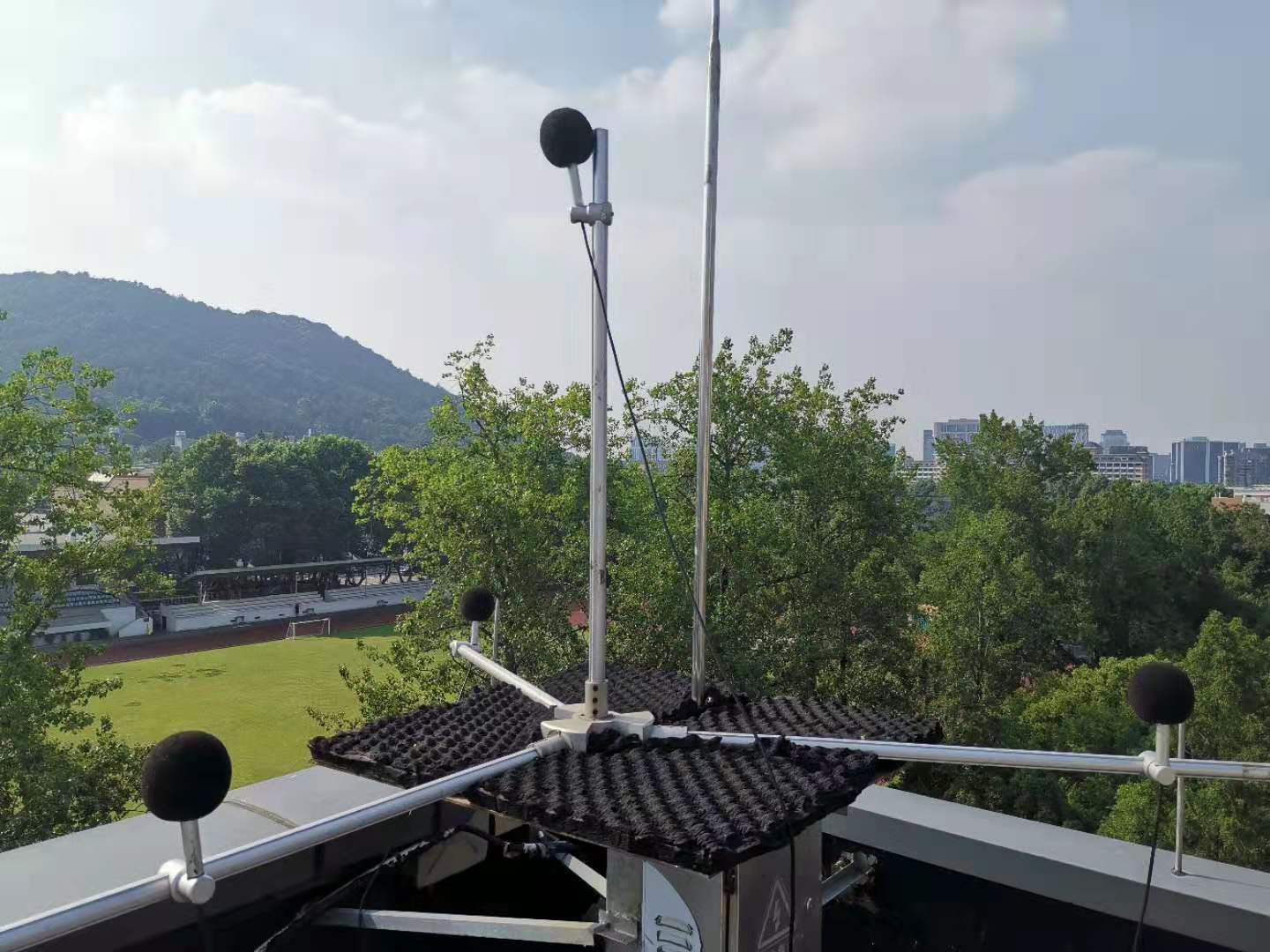}
		\caption{Microphone array}
		\label{microphone_array}
	\end{subfigure}
	\hfill %%
	\begin{subfigure}[t]{0.4\columnwidth}
		\includegraphics[width=\linewidth]{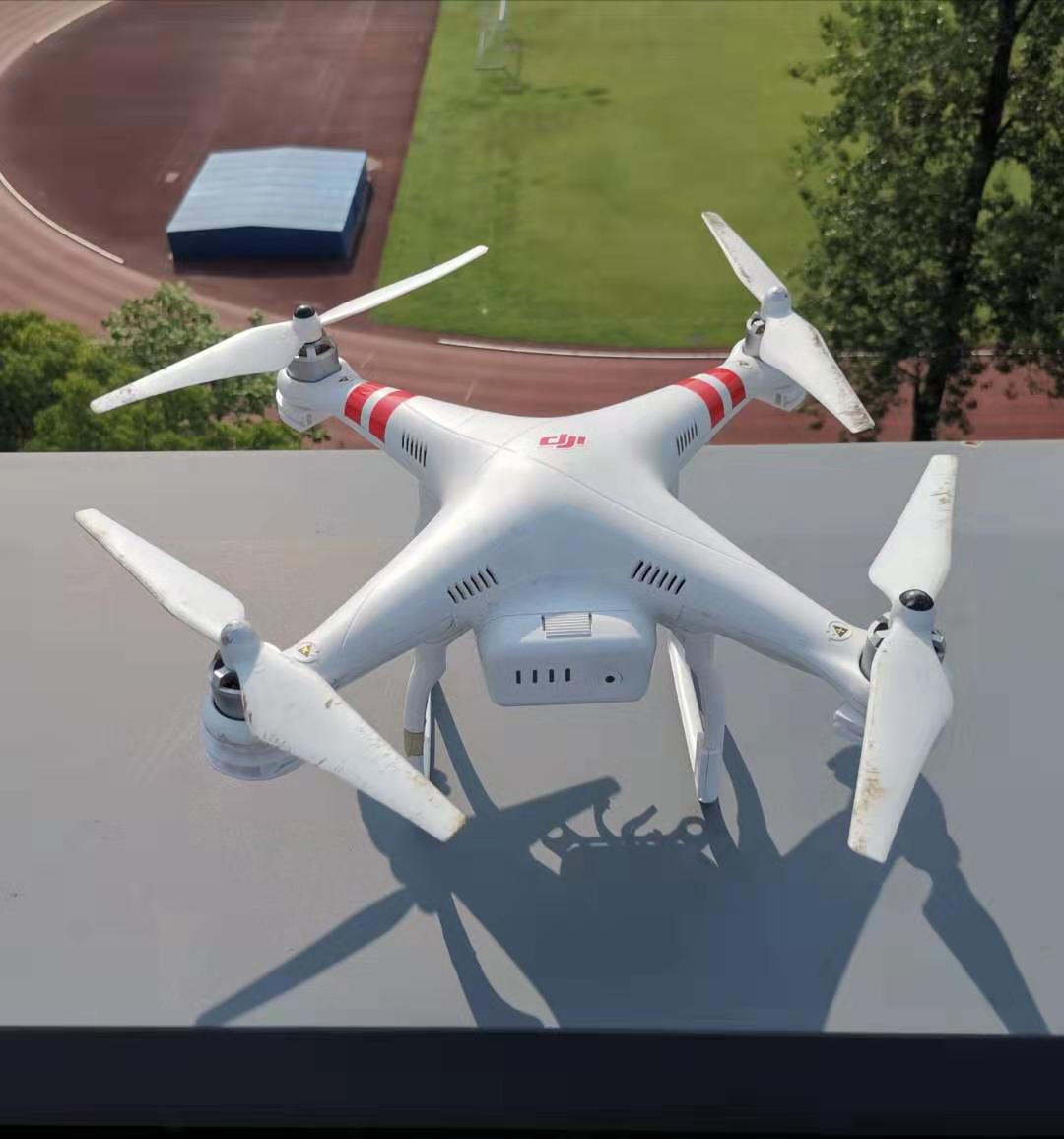}
		\caption{DJI Phantom 2}
		\label{drone_type}
	\end{subfigure}
	\caption{Real-world experiment of acoustic source  detection. The left figure shows the microphone array, where four microphones are deployed at the vertices of a tetrahedron. The right is a figure of a DJI Phantom 2 drone.}
	\label{experiment_equipment}
\end{figure}
We use the NI DAQ device to collect the audio signal, with a sampling frequency set at 25600 Hz. Each microphone detects the drone locally following the procedure below. First, a 36-dimension Mel-frequency cepstral coefficients (MFCC) feature vector is extracted every half a second from the DAQ outputs. Then the MFCC vector is fed into a regression model (a modified version of SVM for regression), which is trained in advance using MATLAB function {\it{fitrsvm}}, to produce a continuous confidence level. Finally, the confidence level is binarized using the threshold rule, by which the local decision is made. In the first step, we collect two sets of signals, one being the drone signals (the drone is flying randomly in a range of 5m to 40m) and the other being background noises. We select an appropriate threshold for each microphone to ensure the four microphones have the same local false alarm probability. The detection and false alarm probabilities of the microphones are listed in TABLE~\ref{train_local_performance}. We see that their detection probabilities vary widely.

\begin{table}[htbp]
	\centering
	\captionsetup{justification=centering}
	\caption{Performances probabilities of the microphones in the first step} \label{train_local_performance}
	\resizebox{0.7\columnwidth}{!}{%
	\begin{tabular}{||c |c c c c||}
		\hline
		 & Mic.1 & Mic.2 & Mic.3 & Mic.4\\ [0.5ex]
		\hline
		p & 0.60 & 0.89 & 0.33 & 0.90\\
		\hline
		q & 0.09 & 0.09 & 0.09 & 0.09\\
		\hline
	\end{tabular}
}
\end{table}

In the second step, we also collect two sets of similar signals and repeat the binarization procedures using the same thresholds in the first step.  Recall that the conditional independence among sensors is required in the analysis, while the microphones are so close to each other that their environmental noises are supposed to be highly correlated. To alleviate the dependence, we randomly disorder each microphone's decision series. The test results are given in TABLE~\ref{test_performance_comparison}.  For the fusion algorithms~\eqref{LR_test_0_one_stage},~\eqref{LR_test_5}, the decision parameters are obtained based on the performance of each microphone obtained in the first step but not the second step. We can see that both fusion algorithms~\eqref{LR_test_0_one_stage},~\eqref{LR_test_5} demonstrate detection performance improvement to some extent, despite inaccuracy rooted in the detection and false alarm probabilities of the individual sensors. Moreover, compared with the decision fusion~\eqref{LR_test_0_one_stage}, our proposed algorithm achieves greater improvement.

\begin{table}[htbp]
	\centering
	\captionsetup{justification=centering}
	\caption{Performance comparison in the second step} \label{test_performance_comparison}
		\resizebox{\columnwidth}{!}{%
\begin{tabular}{||M{1mm} |M{4mm} M{4mm} M{4mm} M{4mm} M{11mm} M{11mm} M{11mm}||}
		\hline
		& Mic.1 & Mic.2 & Mic.3 & Mic.4 & Without memory & Proposed ($t^1=t_0$) & Proposed ($t^1=t_1$)\\ [0.5ex]
		\hline
		p & 0.47 & 0.79 & 0.19 & 0.79 & 0.87 & 0.93 & 0.95\\
		\hline
		q & 0.05 & 0.16 & 0.15 & 0.16 & 0.16 & 0.003 & 0.05\\
		\hline
	\end{tabular}
}
\end{table}

\section{Conclusion and discussion} \label{conclusion}
In this paper, we investigate a multi-stage distributed detection, where the FC is equipped with a one-bit memory storing 
the binary-valued FC's decision at the previous time stage.
We discuss an oracle optimal algorithm under the N-P criterion and explore its structural characterization and limitation of the detection probability in the asymptotic regime. We notice that such an optimal policy is computationally inefficient, and in turn propose a low-complexity fusion algorithm, where the LR test threshold is selected in connection to the message stored in the memory. Our analysis unveils that 
the low-complexity fusion reduces the computational complexity at each stage from $O(4^n)$ to $O(n)$
but promises identical detection performance in the asymptotic regime that the oracle fusion has, in terms of detection probability and the rate of convergence thereof. In future work, a more general multi-stage distributed detection diagram, depicted in Fig.~\ref{future_work}, is of particular interest. In Fig.~\ref{future_work}, a multi-bit memory is available at the FC to store the intermediate variable $x_0^k$. The objective is to design optimal rules $x_0^k=f^k(\bm u^k,x_0^{k-1})$ and $u_0^k=g^k(x_0^k)$ to solve problem~\eqref{N-P_test_original}. We believe that the results in this paper will facilitate the investigation of the multi-bit case.

\begin{figure}[htbp]
	\centering
	\includegraphics[width=0.44\textwidth]{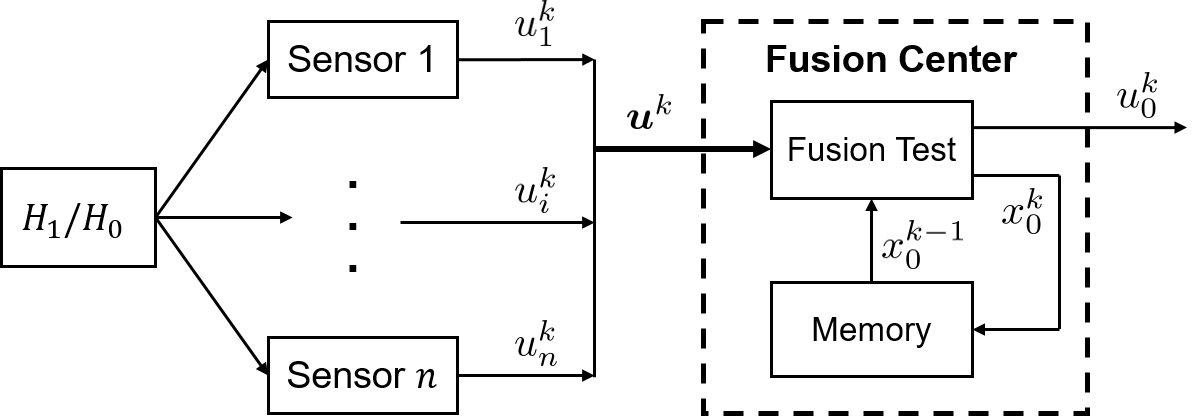}
	\caption{Diagram of a general multi-stage distributed detection with finite memory in a sensor network.}
	\label{future_work}
\end{figure} 

\appendices
\section{Proof of Lemma~\ref{ROC_curve_one_stage}} \label{appendix_ROC_curve_one_stage}
	With no loss of generality, suppose $R_1\geq R_2\geq \dots \geq R_n$. Recall that for $n$ sensors, there are $2^n$ kinds of likelihood ratios ($\Lambda$). If we set the threshold in~\eqref{LR_test_0_one_stage} as $t=\Lambda_{2^n}$, then $q_0=\lambda\prod_{i=1}^{n}q_i$ and $p_0=\lambda\prod_{i=1}^{n}p_i$, hence the slope of the most lower left segment of the ROC curve is
	\begin{equation*}
	l_1=\frac{\prod_{i=1}^{n}p_i}{\prod_{i=1}^{n}q_i}.
	\end{equation*}
	 If we set $t=\Lambda_{2^n-1}$, then $q_0=\prod_{i=1}^{n}q_i+\lambda(1-q_n)\prod_{i=1}^{n-1}q_i$ and $p_0=\prod_{i=1}^{n}p_i+\lambda(1-p_n)\prod_{i=1}^{n-1}p_i$, hence the slope of the second most lower left segment is
	 \begin{equation*}
	 l_2=\frac{(1-p_n)\prod_{i=1}^{n-1}p_i}{(1-q_n)\prod_{i=1}^{n-1}q_i},
	 \end{equation*}
	  and since $q_n<p_n$, we have $l_2<l_1$. Further, if we set $t=\Lambda_{2^n-2}$, then
	  \begin{equation*}
	  l_3=\frac{p_n(1-p_{n-1})\prod_{i=1}^{n-2}p_i}{q_n(1-q_{n-1})\prod_{i=1}^{n-2}q_i},
	  \end{equation*}
	  and since $R_{n-1} \geq R_n$, we have $l_2 \geq l_3$. Similarly, we obtain the slopes of the other segments with $l_1>l_2 \geq \dots \geq l_{2^n-1}>l_{2^n}$, which completes the proof.

\section{Proof of Theorem~\ref{oracle_performance}} \label{appendix_oracle_performance}
\textit{Part I.} Define 
\begin{equation*}
p_0^*:=\frac{\alpha \left(\prod_{i=1}^{n}R_i\right)}{1+\alpha \left(\prod_{i=1}^{n}R_i-1\right)}.
\end{equation*}
First we show that $R_0^1<\prod_{i=1}^{n} R_i$, i.e., $p_0^1<p_0^*$. When $k=1$, by Lemma~\ref{ROC_curve_one_stage}, the FC's ROC consists of $2^n$ segments with non-increasing slopes. The leftmost segment is expressed as
\begin{equation*}
p_0^1=q_0^1 \frac{\prod_{i=1}^{n}p_i}{\prod_{i=1}^{n}q_i} ,~q_0^1 \in (0,\prod_{i=1}^{n}q_i],
\end{equation*}
and the rightmost one is 
\begin{equation*}
p_0^1=1-(1-q_0^1) \frac{\prod_{i=1}^{n}(1-p_i)}{\prod_{i=1}^{n}(1-q_i)} , q_0^1 \in (1-\prod_{i=1}^{n}(1-q_i),1].
\end{equation*}
Then if we extend the two segments, they will intersect at $(p_m,q_m)$, where

\begin{equation*}\label{intersect_point}
(p_m,q_m)=\left( \frac{\prod_{i=1}^{n}R_i-Q_1}{\prod_{i=1}^{n}R_i-1},\frac{Q_2-1}{\prod_{i=1}^{n}R_i-1} \right),
\end{equation*}
where $Q_1=\frac{\prod_{i=1}^{n}p_i}{\prod_{i=1}^{n}q_i}$ and $Q_2=\frac{\prod_{i=1}^{n}(1-q_i)}{\prod_{i=1}^{n}(1-p_i)}$. To show $R_0^1<\prod_{i=1}^{n} R_i$, we consider two cases: 

\noindent (a) Case 1: If $0<\alpha \leq q_m$, let
\begin{equation*}
p_l=\alpha \frac{\prod_{i=1}^{n}p_i}{\prod_{i=1}^{n}q_i},
\end{equation*}
and define
\begin{equation*}
f(\alpha):=\frac{p_l}{1-p_l}\frac{1-\alpha}{\alpha}=\frac{(1-\alpha)(\prod_{i=1}^{n}p_i)}{(\prod_{i=1}^{n}q_i)-\alpha(\prod_{i=1}^{n}p_i)}.
\end{equation*}
Then we have $f'(\alpha)>0$ which leads to $f(\alpha) \leq f(q_m)$, i.e., $f(\alpha) \leq \prod_{i=1}^{n} R_i$ , where the equality holds only when $\alpha=q_m$. When $0<\alpha \leq \prod_{i=1}^{n} q_i$ we have $p_0^1=p_l$ and
\begin{equation*}
R_0^1=\frac{p_0^{1}}{1-p_0^{1}}\frac{1-\alpha}{\alpha}=f(\alpha)<f(q_m)=\prod_{i=1}^{n} R_i.
\end{equation*}
When $\prod_{i=1}^{n} q_i<\alpha \leq q_m$ we have $p_0^1<p_l$ and
\begin{equation*}
R_0^1=\frac{p_0^{1}}{1-p_0^{1}}\frac{1-\alpha}{\alpha}<f(\alpha)\leq f(q_m)=\prod_{i=1}^{n} R_i.
\end{equation*}

\noindent (b) Case 2: If $q_m<\alpha \leq 1$, let
\begin{equation*}
p_u=1-(1-\alpha)\frac{\prod_{i=1}^{n}(1-p_i)}{\prod_{i=1}^{n}(1-q_i)},
\end{equation*}
and define
\begin{equation*}
g(\alpha):=\frac{p_u}{1-p_u}\frac{1-\alpha}{\alpha}=1+\frac{\prod_{i=1}^{n}(1-q_i)-\prod_{i=1}^{n}(1-p_i)}{\alpha\left(\prod_{i=1}^{n}(1-p_i)\right)}.
\end{equation*} Then we have $g'(\alpha)<0$ which leads to $g(\alpha) \leq g(q_m)$, i.e., $g(\alpha) \leq \prod_{i=1}^{n} R_i$ where the equality holds only when $\alpha=q_m$. Noting that $p_0^1=p_u$ when $1-\prod_{i=1}^{n}(1-q_i)\leq \alpha \leq 1$ and $p_0^1<p_u$ when $q_m<\alpha < 1-\prod_{i=1}^{n}(1-q_i)$, similarly we have $R_0^1<g(q_m)=\prod_{i=1}^{n} R_i$. 

Next we show that $R_0^k<R_0^{k+1}<\prod_{i=1}^{n} R_i$ when $R_0^k<\prod_{i=1}^{n} R_i$. Since $R_0^k<\prod_{i=1}^{n} R_i$, by Lemma~\ref{better_performance_2} we have $p_0^k<p_0^{k+1}$, i.e., $R_0^k<R_0^{k+1}$. We note that when $R_0^{k} = \prod_{i=1}^{n} R_i$, by Lemma~\ref{better_performance_2}, $p_0^{k+1}=p_0^{k}=p_0^*$, i.e., $R_0^{k+1} = \prod_{i=1}^{n} R_i$. Then, for a smaller $R_0^{k}<\prod_{i=1}^{n} R_i$, $R_0^{k+1}<\prod_{i=1}^{n} R_i$ holds. 

\textit{Part II.} In the second part, we consider an alternative fusion algorithm, see~\eqref{LR_test_5} and~\eqref{threshold_selection} in Section~\ref{the_proposed_algorithm}, and have that by using this algorithm, $R_0^k$ converges to $\prod_{i=1}^{n} R_i$ (Theorem~\ref{possible_t1}). Due to the optimality of the oracle algorithm,~\eqref{oracle_symptotic_performance2} follows, which completes the proof.
	
\section{Proof of Theorem~\ref{convergence_sped}} \label{appendix_convergence_sped}
	The proof is divided into two parts. In the first part, we show that for both algorithms, we have the iterative form: $p_0^{k}=a^k p_0^{k-1}+b^k$. In the second part, we show that when $k$ is sufficiently large, $a^k$ and $b^k$ of both algorithms converge to the same constant $a$ and $b$ where $a \in (0,1)$. 
	
	\textit{Part I.} For the proposed algorithm, we have 
	\begin{align*}
	p_0^{k}&=P(u_0^{k-1}=1 \mid H_1)p_{0,1}+P(u_0^{k-1}=0 \mid H_1)p_{0,0}\\
	&=p_0^{k-1}p_{0,1}+(1-p_0^{k-1})p_{0,0}\\
	&=(p_{0,1}-p_{0,0})p_0^{k-1}+p_{0,0}.
	\end{align*} 
	For the oracle optimal algorithm, recall that at each stage we first sort the likelihood ratios as $\Lambda_1^k \le\Lambda_2^k \le \dots \le \Lambda_{2^{n+1}}^k$. Suppose $t^k=\Lambda_j^k$ for some $1\leq j \leq 2^{n+1}$, then we obtain 
	\begin{equation*}
	p_0^{k}=\sum_{j<i \leq 2^{n+1}} P(\Lambda_i^{k}|H_1)+\lambda^k \sum_{i=j} P(\Lambda_i^k|H_1),
	\end{equation*}  
	which can be reduced to 
	\begin{equation*}
	p_0^{k}=a^k p_0^{k-1}+b^k,
	\end{equation*}
	for some $a^k$ and $b^k$. We remark that $a^k$ and $b^k$ may change in different stages since the order of the likelihood ratios may change with the increase of $p_0^k$. 
	
	\textit{Part II.} When $k$ is large enough, for the oracle optimal algorithm, the order of the likelihood ratios $\Lambda^k(u_0^{k-1},u_1^k,\ldots,u_n^k)$ is fixed as follows:
	\begin{align*}
	\Lambda^k(1,0,\ldots,0)<\Lambda^k(0,1,\ldots,1)<\Lambda^k(1,0,\ldots,0,1) \\
	\leq \cdots<\Lambda^k(1,1,\ldots,1),
	\end{align*}
	where the first inequality holds because $R_0^{k}<\prod_{i=1}^{n} R_i$. Hence we have $\Lambda^k_{2^n}=\Lambda^k(1,0,\ldots,0)$ and $\Lambda^k_{2^n+1}=\Lambda^k(0,1,\ldots,1)$. Next we consider two cases:
	
	\noindent (a) Case 1: $\prod_{i=1}^{n} q_i<\frac{\alpha}{1-\alpha} \prod_{i=1}^{n} (1-q_i)$, i.e., $1<Q_4$. In this case, for the oracle optimal algorithm, we have $P(\Lambda^k_{2^n} \mid H_0)>P(\Lambda^k_{2^n+1} \mid H_0)$. In addition, note that 
	\begin{equation*}
	\sum_{u_1^k} \cdots \sum_{u_n^k} P\left(\Lambda^k(1,u_1^k,\ldots,u_n^k)\mid H_0 \right)=\alpha.
	\end{equation*}
	Hence $t^k=\Lambda^k_{2^n}$ and 
	\begin{align*}
	\lambda^k= & \frac{P(\Lambda^k_{2^n} \mid H_0)-P(\Lambda^k_{2^n+1} \mid H_0)}{P(\Lambda^k_{2^n} \mid H_0)} \\
	& =\frac{\alpha \prod_{i=1}^{n}(1-q_i)-(1-\alpha) \prod_{i=1}^{n}q_i}{\alpha \prod_{i=1}^{n}(1-q_i)},
	\end{align*}
	which lead to 
	\begin{align}
	p_0^k & =\left(1-(1-\lambda^k)\prod_{i=1}^{n}(1-p_i) \right) p_0^{k-1}+(1-p_0^{k-1})\prod_{i=1}^{n}p_i \notag\\
	& =a p_0^{k-1}+b, \label{a_case1}
	\end{align}
	where $a=1-\left(1+Q_3 \right) \prod_{i=1}^{n} p_i$ and $b=\prod_{i=1}^{n}p_i$. For the proposed algorithm, we have 
	\begin{equation*}
	q_{0,0} =\min \left\{\prod_{i=1}^{n} q_i,\frac{\alpha}{1-\alpha} \prod_{i=1}^{n} (1-q_i)\right\}=\prod_{i=1}^{n} q_i,
	\end{equation*}
	and 
	\begin{equation*}
	q_{0,1}=1+\prod_{i=1}^{n} q_i-\frac{\prod_{i=1}^{n} q_i}{\alpha}.
	\end{equation*}
	As illustrated in Fig.~\ref{appendix_1}, in this case, $(q_{0,0},p_{0,0})$ locates at the intersection of the first segment and the second segment of ROC curve and $(q_{0,1},p_{0,1})$ falls on the last segment. As a consequence, we have 
	\begin{align*}
	p_{0,1}-p_{0,0} & =l_{2^n}q_{0,1}+(1-l_{2^n})-l_1q_{0,0} \\
	& =\frac{\prod_{i=1}^{n}(1-p_i)}{\prod_{i=1}^{n}(1-q_i)} \left(q_{0,1}-1 \right)+1-\frac{\prod_{i=1}^{n}p_i}{\prod_{i=1}^{n}q_i}q_{0,0},
	\end{align*}
	and $p_{0,0}=\prod_{i=1}^{n} p_i$ which coincide with $a$ and $b$ in~\eqref{a_case1}. 
	
	\noindent (b) Case 2: $\prod_{i=1}^{n} q_i \geq \frac{\alpha}{1-\alpha} \prod_{i=1}^{n} (1-q_i)$, i.e., $1 \geq Q_4$. In this case, for the oracle optimal algorithm, we have $P(\Lambda^k_{2^n} \mid H_0)\leq P(\Lambda^k_{2^n+1} \mid H_0)$, hence $t^k=\Lambda^k_{2^n+1}$ and 
	\begin{equation*}
	\lambda^k= \frac{P(\Lambda^k_{2^n} \mid H_0)}{P(\Lambda^k_{2^n+1} \mid H_0)} =\frac{\alpha \prod_{i=1}^{n}(1-q_i)}{(1-\alpha) \prod_{i=1}^{n}q_i},
	\end{equation*}
	which lead to 
	\begin{align}
	p_0^k & =\left(1-\prod_{i=1}^{n}(1-p_i) \right) p_0^{k-1}+\lambda^k(1-p_0^{k-1})\prod_{i=1}^{n}p_i \notag \\
	& =a p_0^{k-1}+b, \label{a_case2}
	\end{align}
	where $a=1-\left(1+\frac{1}{Q_3} \right) \prod_{i=1}^{n} (1-p_i)$ and $b=\lambda^k \prod_{i=1}^{n}p_i$. For the proposed algorithm, we have 
	\begin{equation*}
	q_{0,0} =\min \left\{\prod_{i=1}^{n} q_i,\frac{\alpha}{1-\alpha} \prod_{i=1}^{n} (1-q_i)\right\}=\frac{\alpha}{1-\alpha} \prod_{i=1}^{n} (1-q_i),
	\end{equation*}
	and 
	\begin{equation*}
	q_{0,1}=1-\prod_{i=1}^{n} (1-q_i).
	\end{equation*}
	As illustrated in Fig.~\ref{appendix_2}, in this case, $(q_{0,1},p_{0,1})$ locates at the intersection of the last two segments of ROC curve and $(q_{0,0},p_{0,0})$ falls on the first segment. As a consequence, we have 
	\begin{align*}
	p_{0,1}-p_{0,0} & =l_{2^n}q_{0,1}+(1-l_{2^n})-l_1q_{0,0} \\
	& =\frac{\prod_{i=1}^{n}(1-p_i)}{\prod_{i=1}^{n}(1-q_i)} \left(q_{0,1}-1 \right)+1-\frac{\prod_{i=1}^{n}p_i}{\prod_{i=1}^{n}q_i}q_{0,0},
	\end{align*}
	and $p_{0,0}=\lambda^k \prod_{i=1}^{n}p_i$ which coincide with $a$ and $b$ in~\eqref{a_case2} and complete the proof. 
		\begin{figure}[htbp]
		\captionsetup{singlelinecheck = false, justification=justified}
		\captionsetup[subfigure]{justification=centering}
		\begin{subfigure}[t]{0.49\columnwidth}
			\includegraphics[width=\linewidth]{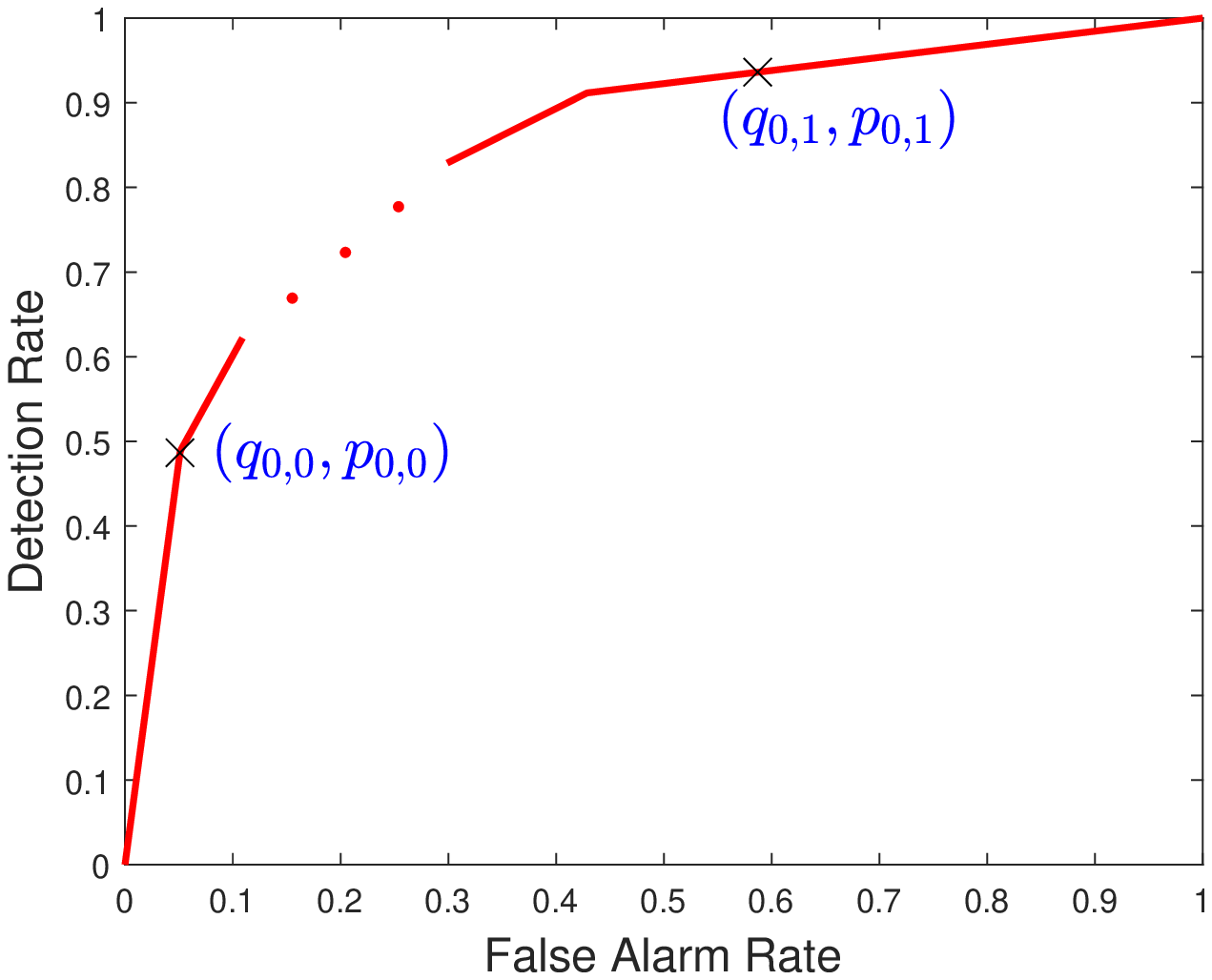}
			\caption{Case 1}
			\label{appendix_1}
		\end{subfigure}
		\hfill %%
		\begin{subfigure}[t]{0.49\columnwidth}
			\includegraphics[width=\linewidth]{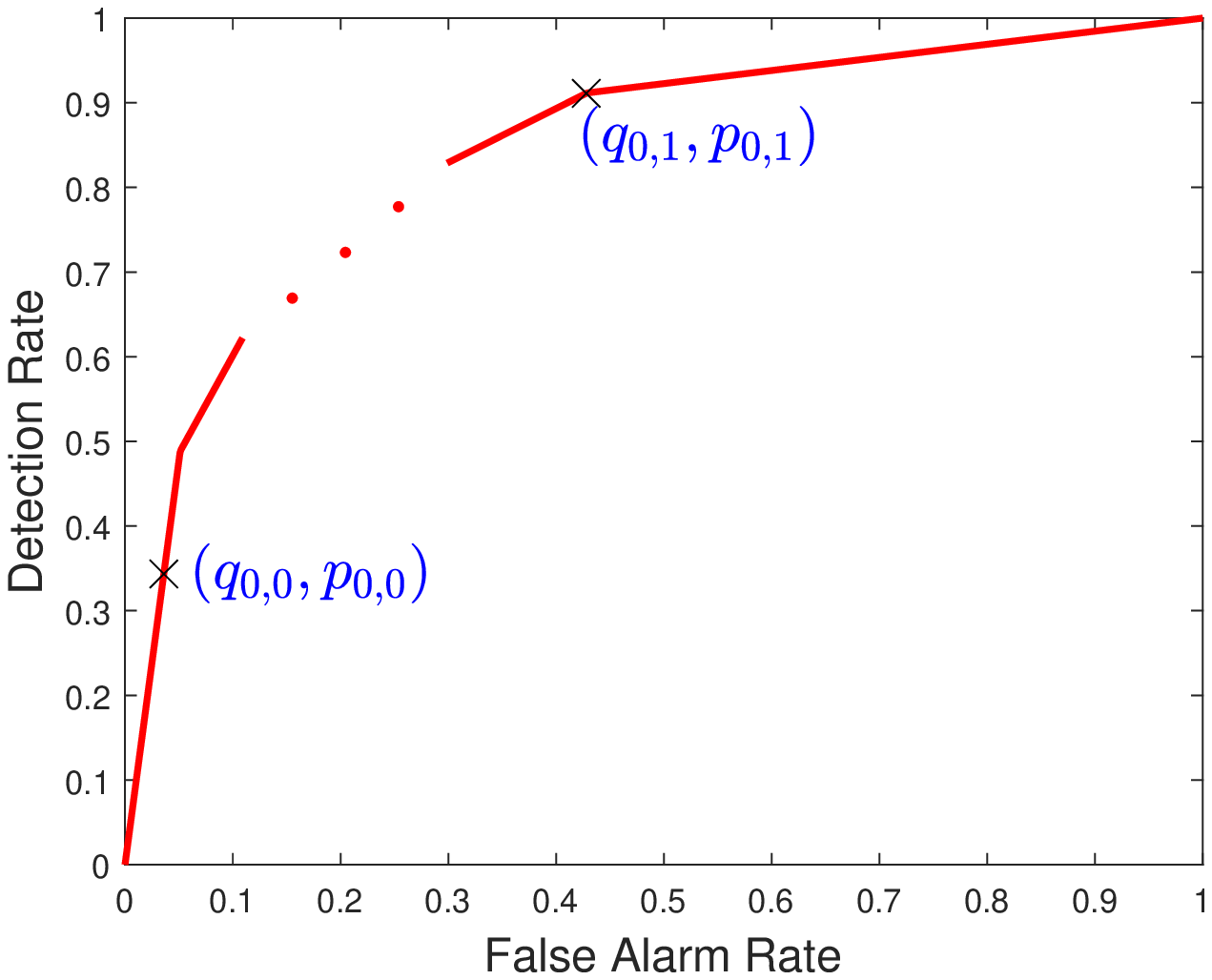}
			\caption{Case 2}
			\label{appendix_2}
		\end{subfigure}
		\caption{Thresholds of the proposed algorithm. For case 1, $(q_{0,0},p_{0,0})$ locates at the intersection of the first segment and the second segment of ROC curve. While for case 2, $(q_{0,1},p_{0,1})$ locates at the intersection of the last two segments.}
		\label{appendix}
	\end{figure}
%\section{}
%\section*{Acknowledgment}

\ifCLASSOPTIONcaptionsoff
  \newpage
\fi



%\bibliographystyle{IEEEtran}
\small
\bibliographystyle{IEEEtran}
\bibliography{Reference}


%\end{IEEEbiography}
%% insert where needed to balance the two columns on the last page with
%%\newpage


%\vfill


\end{document}